\documentclass[journal]{IEEEtran}

\ifCLASSINFOpdf
\usepackage[pdftex]{graphicx}
\DeclareGraphicsExtensions{.pdf,.jpeg,.png}
\else
\usepackage[dvips]{graphicx}
\fi
\usepackage{epstopdf}
\usepackage{amssymb,amsmath,amsthm}
\usepackage{wasysym}
\usepackage{algorithm}
\usepackage{algorithmic}
\usepackage{array}
\usepackage{cite}
\usepackage{color}
\usepackage{url}

\usepackage{epsfig,latexsym}
\usepackage{flushend}
\usepackage{verbatim}
\usepackage{amsopn}
\usepackage{booktabs}

\usepackage{stfloats}
\usepackage{enumerate}
\usepackage{hyperref}
\usepackage{subfigure}
\usepackage{caption}
\usepackage{bm}
\usepackage{mathtools}

\newtheorem{theorem}{Theorem}

\renewcommand{\baselinestretch}{0.95}

\begin{document}
	
	\title{3-D Trajectory Optimization for Robust Direction Sensing in Movable Antenna Systems}
	
	\author{{Wenyan Ma, \IEEEmembership{Member, IEEE}, Lipeng Zhu, \IEEEmembership{Senior Member, IEEE}, Xiaodan Shao, \IEEEmembership{Member, IEEE}, and  Rui Zhang, \IEEEmembership{Fellow, IEEE}}
		\thanks{
			\textit{(Corresponding authors: Lipeng Zhu; Xiaodan Shao.)}}
		\thanks{W. Ma and R. Zhang are with the Department of Electrical and Computer Engineering, National University of Singapore, Singapore 117583 (e-mail: wenyan@u.nus.edu, elezhang@nus.edu.sg).}
		\thanks{L. Zhu is with the State Key Laboratory of CNS/ATM and the School of Interdisciplinary Science, Beijing Institute of Technology, Beijing 100081, China (e-mail: zhulp@bit.edu.cn).}
		\thanks{X. Shao is with the Department of Electrical and Computer Engineering, University of Waterloo, Waterloo, ON N2L 3G1, Canada (e-mail: x6shao@uwaterloo.ca).}
		}
	\maketitle
	
	\begin{abstract}
			This paper presents a novel wireless sensing system where a movable antenna (MA) continuously moves and receives sensing signals within a three-dimensional (3-D) region to enhance sensing performance compared with conventional fixed-position antenna (FPA)-based sensing. We show that the performance of direction vector estimation for a target is fundamentally related to the 3-D MA trajectory in terms of the mean square angular error lower-bound (MSAEB), which is adopted as a coordinate-invariant performance metric. In particular, the  closed-form expression of the MSAEB is derived as a function of the trajectory covariance matrix. Theoretical analysis shows that two-dimensional (2-D) antenna movement suffers from performance divergence for target direction close to the endfire direction of the 2-D MA plane, whereas 3-D movement can achieve isotropic sensing performance over the entire angular region. To achieve robust sensing performance, we formulate a min–max optimization problem to minimize the maximum (worst-case) MSAEB over a given continuous angular region wherein the target is located. An efficient successive convex approximation (SCA) algorithm is developed to optimize the 3-D MA trajectory and obtain a locally optimal solution. Numerical results demonstrate that the proposed 3-D MA sensing scheme is able to significantly reduce the worst-case mean square angular error (MSAE) compared with conventional arrays with FPAs and MA systems with 2-D movement only, thus achieving more accurate and robust direction estimation over the entire angular region.
	\end{abstract}
	
	
	\section{Introduction}
	The sixth-generation (6G) mobile communication systems are expected to support a wide range of location-critical services, including autonomous driving, robotic navigation, and drone coordination \cite{jiang2021the}. These applications require high-resolution sensing capabilities beyond conventional metrics of communication throughput and link reliability. This paradigm shift has motivated the development of integrated sensing and communication (ISAC), a framework that jointly exploits spectral and hardware resources to enable both communication and sensing. In future 6G networks, wireless sensing (e.g., detection and parameter estimation of physical targets) is anticipated to evolve from an auxiliary function into a fundamental network service.
	
	To achieve high angular resolution and effective beamforming for sensing enhancement, conventional radars or base stations (BSs) typically require large-scale antenna arrays \cite{mailloux2005phased}. However, the increasing number of antenna elements and associated radio frequency (RF) chains incur substantial hardware cost and power consumption, limiting the feasibility of cost-efficient sensing solutions. To reduce the hardware cost, sparse antenna configurations have been proposed to maintain angular resolution with fewer antenna elements by effectively expanding the aperture \cite{roberts2011sparse}. However, these designs generally rely on fixed-position antennas (FPAs), which suffer from obstinate sidelobes and lack of the flexibility to dynamically adapt their geometry to varying sensing requirements. Moreover, in both large-scale and sparse arrays, FPAs cannot fully harness the spatial degrees of freedom (DoFs) available within the transceiver’s deployment region.
	
	To overcome the inherent limitations of FPA-based sensing, recent works have explored a new sensing architecture based on movable antenna (MA) \cite{zhu2023MAMag,zhu2025tutorial}, where the transmitter and/or receiver can dynamically adjust the antenna's position. This additional DoF enables substantial performance gains over conventional FPA arrays, with the same or even smaller number of antenna elements. In particular, enlarging the antenna movement region effectively increases the effective aperture, thereby improving angular resolution. Moreover, by optimizing the MA trajectory, the mutual correlation between steering vectors corresponding to different directions can be reduced, which suppresses sidelobes/interference and alleviates estimation ambiguity. In practice, the MA trajectory can be pre-programmed for specific sensing tasks or adaptively optimized in real time to accommodate dynamic sensing requirements.
	
	Research on MA-enabled  wireless systems dates back to 2009, when adjusting a single antenna’s position within a given region was shown to achieve significant diversity gains for wireless communication \cite{zhao2009single}. In recent years, MAs have been regarded as a transformative paradigm in modern multi-antenna or multiple-input multiple-output (MIMO) systems, demonstrating superior performance over FPAs in various scenarios. Studies in \cite{zhu2022MAmodel,mei2024movable,ning2024movable,tang2024secure} showed that optimizing antennas' positions over a wavelength-scale movement region can significantly enhance the received signal-to-noise ratio (SNR) under diverse channel conditions. Beyond SNR improvements, extensive research has focused on multiuser interference mitigation via joint antenna position and beamforming optimization \cite{zhu2023MAmultiuser,wu2023movable,qin2024antenna,cheng2023sum,yang2024flexible,hu2024power,li2024minimizing}. Furthermore, the spatial multiplexing performance of MA-enabled  MIMO systems and the corresponding channel estimation methods have been thoroughly investigated in \cite{ma2022MAmimo,chen2023joint,yeyuqi2023fluid,ma2023MAestimation,xiao2023channel}. The benefits of MAs have also been validated in multi-beamforming, satellite communications, and near-field scenarios \cite{zhu2023MAarray,ma2024multi,ZhuLP_satellite_MA,zhu2024nearfield}. As an extension, six-dimensional MA (6DMA) architectures were proposed in \cite{shao20246DMA,shao2024discrete,shao2024Mag6DMA,shao2024exploiting} to further incorporate three-dimensional (3-D) rotational DoFs in addition to spatial translation.
	
	The advantages of MAs have also been validated in wireless sensing systems. Early experimental studies demonstrated that adjusting antennas' positions can enhance radar imaging quality and improve target localization accuracy \cite{zhuravlev2015experi,hinske2008using}. However, while these pioneering efforts verified the practical benefits of antenna movement, they did not reveal the theoretical connection between antenna movement and sensing performance. To bridge this gap, subsequent works have optimized MA array geometries by minimizing the Cramér–Rao bound (CRB) for angle-of-arrival (AoA) estimation in both far-field and near-field scenarios \cite{ma2024MAsensing,chen2025MAISACopt,wang2025MAnearsensing,mao2025movable}. These studies derived closed-form CRB expressions as explicit functions of antennas' positions to facilitate antenna position optimization. The framework was further generalized to 6DMA architectures via joint antenna position and orientation optimization, enabling sum CRB minimization for targets distributed across multiple spatial sectors \cite{shao2024exploiting}. These studies designed the MA array geometry for a specific sensing objective and the geometry remains fixed throughout the sensing process. Moreover, the work in \cite{ma2025movabletra} exploited the joint time-spatial DoFs enabled by antenna movement. By continuously moving the antenna while receiving echo signals, a large-aperture virtual array can be synthesized, thereby achieving high angular resolution without incurring grating lobes. However, this study only considered one-dimensional (1-D) and two-dimensional (2-D) antenna movement. When the target direction is close to the endfire of the antenna movement axis or plane, the resulting effective aperture of the virtual MA array becomes significantly reduced, leading to degraded sensing performance. Consequently, 1-D and 2-D antenna movement cannot guarantee robust sensing performance for arbitrary target directions in 3-D space.
	
	To overcome this limitation, we investigate in this paper an MA-enabled  sensing system that fully exploits the spatial DoFs enabled by 3-D antenna movement. In contrast to existing studies that primarily optimize static MA array geometries \cite{ma2024MAsensing,chen2025MAISACopt,wang2025MAnearsensing,shao2024exploiting}, we consider the scenario where the antenna continuously moves while receiving echo signals. By optimizing the MA trajectory in a 3-D region, a large continuous virtual aperture is synthesized, thereby enhancing sensing accuracy. Furthermore, unlike the 1-D or 2-D trajectory designs considered in \cite{ma2025movabletra}, which may experience effective aperture collapse for certain target directions, we consider the 3-D trajectory optimization to ensure robust sensing performance for arbitrary directions over the entire angular region. The main contributions of this work are summarized as follows:
	
	\begin{itemize}
		\item First, to establish a robust and coordinate-invariant performance metric, we consider the mean square angular error (MSAE) for target direction vector estimation. We derive a closed-form expression for the MSAE lower-bound (MSAEB), which is characterized as a function of the covariance matrix of the 3-D MA trajectory and the target direction vector. To ensure robust sensing performance across different directions, we formulate a maximum (worst-case) MSAEB minimization problem over a continuous angular region via 3-D MA trajectory optimization.
		\item Next, we present a comprehensive theoretical analysis of 3-D MA-enabled  sensing. We prove that antenna movement along the target direction provides no estimation performance gain, and rigorously show that 2-D planar trajectories inevitably suffer from performance divergence for target directions near the endfire direction of the movement plane. Moreover, we characterize the isotropy property of the MSAEB, proving that constant sensing performance across all directions is achieved if and only if the MA trajectory's covariance matrix is a scaled identity matrix, based on which practical isotropic 3-D MA trajectories are designed.
		\item Furthermore, we formulate a MA trajectory optimization problem to minimize the worst-case MSAEB over a given continuous angular region wherein the target is located, subject to the practical constraints on the MA's maximum speed and 3-D movement region. To deal with the non-convex and intractable objective function, we discretize the angular region and thereby develop an efficient algorithm based on successive convex approximation (SCA). For the special case of a single target direction in the angular region, a simplified algorithm is developed, which only needs to optimize the MA trajectory in the 2-D plane orthogonal to the target direction vector.
		\item Finally, numerical results are provided to evaluate the proposed 3-D antenna movement solutions for robust direction sensing. The results show that the proposed scheme significantly outperforms benchmark solutions based on FPAs or MAs with 2-D movement only in terms of estimation MSAE. Furthermore, the optimized 3-D MA trajectories effectively enlarge the virtual spatial aperture, thereby ensuring robust and high-accuracy sensing over the entire angular region.
	\end{itemize}
	
	The rest of this paper is organized as follows. Section II introduces the system model and derives the closed-form MSAEB for target direction vector estimation. Section III provides a detailed performance analysis. In Section IV, an SCA-based optimization algorithm is proposed to solve the 3-D trajectory optimization problem. Section V presents numerical results to validate the proposed scheme, and Section VI concludes the paper.
	
	\textit{Notations}: Boldface lowercase and uppercase letters represent vectors and matrices, respectively. The operators $(\cdot)^{\mathsf *}$, $(\cdot)^{\mathsf T}$, and $(\cdot)^{\mathsf H}$ represent the conjugate, transpose, and the conjugate transpose (Hermitian), respectively. $\mathbb{C}^{P \times Q}$ and $\mathbb{R}^{P \times Q}$ denote the spaces of $P \times Q$ complex-valued and real-valued matrices, respectively. The $p$th entry of vector $\bm{a}$ is denoted by $\bm{a}[p]$, and the entry of matrix $\bm{A}$ in its $p$th row and $q$th column is denoted by $\bm{A}[p,q]$. The $N \times N$ identity matrix is represented by $\bm{I}_N$. $\text{Tr}(\bm{A})$ denotes the trace of matrix $\bm{A}$. Finally, $\|\bm{a}\|_2$ represents the $2$-norm of vector $\bm{a}$.

	\begin{figure}[!t]
		\centering
		\includegraphics[width=85mm]{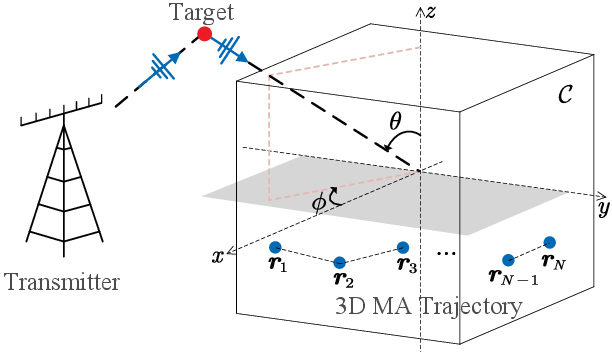}
		\caption{System model for the considered MA-enabled  sensing system.}
		\label{model_3-D}
	\end{figure}
	
	\begin{figure}[!t]
		\centering
		\includegraphics[width=65mm]{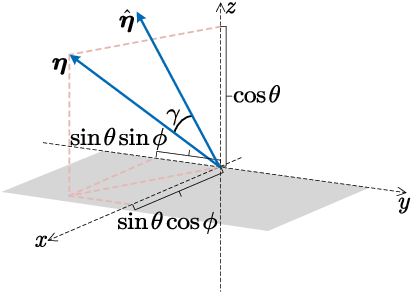}
		\caption{Illustration of the direction vector $\bm{\eta}$ and its angular error $\gamma$.}
		\label{3-D_coordinate}
	\end{figure}
	
	\section{MSAEB Characterization for 3-D MA Sensing}
	
	\subsection{System Model}
	As illustrated in Fig.~\ref{model_3-D}, we consider a bistatic sensing system in which the receiver employs a single MA moving within a 3-D region $\mathcal{C}$. For target direction estimation, the sensing transmitter repeatedly broadcasts probing signals, whereas the MA moves continuously and collects the target-reflected signals over $N$ discrete snapshots. Specifically, the MA receives one reflected signal at each snapshot, with the sampling period denoted by $T_s$. Thus, the total time for sensing and antenna movement is $T = N T_s$. We define the initial sensing time as $t_1 = 0$, at which the MA is located at position $\bm{r}_1 = [x_1, y_1, z_1]^{\mathsf T} \in \mathcal{C}$. The time of the $n$th snapshot ($n = 1, \ldots, N$) is denoted by $t_n = (n-1)T_s$. Let $\bm{r}_n = [x_n, y_n, z_n]^{\mathsf T} \in \mathcal{C}$ represent the MA's position at time $t_n$. We collect all MA's positions over the $N$ snapshots into the antenna position matrix (APM) $\bm{R} = [\bm{r}_1, \bm{r}_2, \ldots, \bm{r}_N] \in \mathbb{R}^{3\times N}$. Owing to the limited movement speed of practical MAs and the short duration of each probing signal, the MA's position is assumed to be quasi-static within each snapshot. The MA's velocity between positions $\bm{r}_n$ and $\bm{r}_{n+1}$ is denoted by $\bm{v}_n = [v_n^x, v_n^y, v_n^z]^{\mathsf T}$ ($n = 1, \ldots, N-1$), which is subject to a maximum speed constraint $v^{\rm m} > 0$, i.e., $\|\bm{v}_n\|_2 \le v^{\rm m}$. Accordingly, the MA's position at snapshot $n$ ($n = 2, \ldots, N$) can be represented as a function of the initial position $\bm{r}_1$ and the antenna velocity matrix (AVM) $\bm{V} = [\bm{v}_1, \dots, \bm{v}_{N-1}] \in \mathbb{R}^{3\times (N-1)}$:
	\begin{align}
		\bm{r}_n &= \bm{r}_{n-1} + \bm{v}_{n-1}(t_n - t_{n-1}) \notag\\
		&= \bm{r}_1 + T_s \sum_{m=1}^{n-1} \bm{v}_m.
	\end{align}
	We consider that the propagation between the target and the receiver is dominated by the line-of-sight (LoS) component, which remains constant over the $N$ sensing snapshots. Since the distance between the target and the receiver is significantly larger than the antenna movement region size, the far-field condition is adopted for the target–receiver channel. As shown in Fig.~\ref{3-D_coordinate}, the physical elevation and azimuth AoAs of the target are denoted by $\theta \in [0, \pi]$ and $\phi \in [0, 2\pi]$, respectively. Accordingly, the direction vector of the target is denoted by
	\begin{align}
		\bm{\eta}\triangleq[\sin\theta \cos\phi, \sin\theta \sin\phi, \cos\theta]^{\mathsf T}.
	\end{align}
	Then, the steering vector of the MA over $N$ snapshots can be expressed as a function of the APM $\bm{R}$ and the direction vector $\bm{\eta}$:
	\begin{equation}
		\bm{\alpha}(\bm{R},\bm{\eta}) \triangleq \left[ e^{j\frac{2\pi}{\lambda}\bm{\eta}^{\mathsf T}\bm{r}_1}, \ldots, e^{j\frac{2\pi}{\lambda}\bm{\eta}^{\mathsf T}\bm{r}_N} \right]^{\mathsf T} \in \mathbb{C}^{N\times1},
	\end{equation}
	where $\lambda$ denotes the carrier wavelength. Let $\beta$ represent the complex-valued channel gain of the LoS path measured at the origin of the antenna movement region $\mathcal{C}$. Thus, the LoS channel vector between the target and the receiver is given by
	\begin{equation}
		\bm{h}(\bm{R},\bm{\eta}) = \beta \bm{\alpha}(\bm{R},\bm{\eta}).	
	\end{equation}
	We further assume that the sensing transmitter continuously radiates a fixed omnidirectional probing waveform. Consequently, by stacking the received signals over all $N$ snapshots, the received signal vector can be written as
	\begin{equation}
		\bm{y} = \bm{h}(\bm{R},\bm{\eta})s + \bm{n},
	\end{equation}
	where $s$ denotes the target-reflected signal with average power $P = \mathbb{E}\{|s|^2\}$, and $\bm{n} \sim \mathcal{CN}(0, \sigma^2 \bm{I}_N)$ represents the additive white Gaussian noise (AWGN) with $\sigma^2$ being the average noise power.

	\subsection{Direction Vector Estimation}
	For a given MA trajectory $\bm{R}$, the direction vector is estimated using the maximum likelihood estimation (MLE) method. Specifically, the unknown parameters $\tilde{\beta} \triangleq \beta s$ and direction vector $\bm{\eta}$ are jointly estimated by solving the following least-squares problem,
	\begin{align}\label{hatbetaeta}
		\left(\hat{\beta},\hat{\bm{\eta}}\right) = \arg\min_{\bar{\beta},\bar{\bm{\eta}}} \|\bm{y} - \bar{\beta}\bm{\alpha}(\bm{R},\bar{\bm{\eta}})\|_2^2.
	\end{align}
	For any given candidate value of $\bm{\eta}$, the optimal estimation of $\tilde{\beta}$ is available in closed-form and given by
	\begin{align}\label{hatbeta}
		\hat{\beta} = \frac{\bm{\alpha}(\bm{R},\bm{\eta})^{\mathsf H} \bm{y}}{\|\bm{\alpha}(\bm{R},\bm{\eta})\|_2^2}.
	\end{align}
	By substituting \eqref{hatbeta} back into \eqref{hatbetaeta}, we obtain
	\begin{align}
		&\|\bm{y} - \hat{\beta}\bm{\alpha}(\bm{R},\bm{\eta})\|_2^2 \\
		=& \left(\bm{y} - \hat{\beta}\bm{\alpha}(\bm{R},\bm{\eta})\right)^{\mathsf H} \left(\bm{y} - \hat{\beta}\bm{\alpha}(\bm{R},\bm{\eta})\right) \notag\\
		=&\|\bm{y}\|_2^2 + |\hat{\beta}|^2 \|\bm{\alpha}(\bm{R},\bm{\eta})\|_2^2  - 2\Re\left(\hat{\beta}\bm{y}^{\mathsf H} \bm{\alpha}(\bm{R},\bm{\eta})\right) \notag\\
		=&\|\bm{y}\|_2^2 - \frac{\left|\bm{y}^{\mathsf H} \bm{\alpha}(\bm{R},\bm{\eta})\right|^2}{\|\bm{\alpha}(\bm{R},\bm{\eta})\|_2^2} \notag\\
		=&\|\bm{y}\|_2^2 - \frac{1}{N}\left|\bm{y}^{\mathsf H} \bm{\alpha}(\bm{R},\bm{\eta})\right|^2. \notag
	\end{align}
	Since $\|\bm{y}\|_2^2$ is independent of $\bm{\eta}$, minimizing $\|\bm{y} - \hat{\beta}\bm{\alpha}(\bm{R},\bm{\eta})\|_2^2$ is equivalent to maximizing $\left|\bm{y}^{\mathsf H} \bm{\alpha}(\bm{R},\bm{\eta})\right|^2$, i.e.,
	\begin{align}\label{MLE1-D}
		\hat{\bm{\eta}} &= \arg\min_{\substack{\bar{\theta}\in[0,\pi]\\ \bar{\phi}\in[0,2\pi]}} \|\bm{y}\|_2^2 - \frac{1}{N}\left|\bm{y}^{\mathsf H} \bm{\alpha}(\bm{R},\bar{\bm{\eta}})\right|^2 \\
		&= \arg\max_{\substack{\bar{\theta}\in[0,\pi]\\ \bar{\phi}\in[0,2\pi]}} \left|\bm{y}^{\mathsf H} \bm{\alpha}(\bm{R},\bar{\bm{\eta}})\right|^2, \notag
	\end{align}
	which can be efficiently obtained by exhaustively searching over the candidate direction vector $\bar{\bm{\eta}} = [\sin\bar{\theta} \cos\bar{\phi}, \sin\bar{\theta} \sin\bar{\phi}, \cos\bar{\theta}]^{\mathsf T}$,	where $\bar{\theta}\in[0,\pi]$ and $\bar{\phi}\in[0,2\pi]$. As shown in Fig.~\ref{3-D_coordinate}, since $\|\bm{\eta}\|_2 = \|\hat{\bm{\eta}}\|_2 = 1$, the angular error of the estimated direction vector $\hat{\bm{\eta}}$ is
	\begin{align}
		\gamma \triangleq \arccos\left(\frac{\bm{\eta}^{\mathsf T}\hat{\bm{\eta}}}{\|\bm{\eta}\|_2\|\hat{\bm{\eta}}\|_2}\right) =\arccos\left(\bm{\eta}^{\mathsf T}\hat{\bm{\eta}}\right).
	\end{align}
	To characterize the direction vector estimation error, we adopt the MSAE, defined as $\mathbb{E}\{\gamma^2\}$, as the performance metric. First, unlike the component-wise mean square errors (MSEs) of the physical elevation and azimuth AoAs $\theta$ and $\phi$, the angular error $\gamma$ avoids parameter coupling near the poles (i.e., $\theta \approx 0$ or $\theta \approx \pi$). Specifically, when $\theta=0$, the direction vector reduces to $\bm{\eta}\equiv[0,0,1]^{\mathsf T}$, which is independent of $\phi$, causing the CRB for estimating $\phi$ to diverge. Second, the angular error $\gamma$ eliminates AoA-wrapping ambiguities. Specifically, when $\phi\approx 0$, an estimate of $\phi$ close to $2\pi$ may lead to a large MSE in $\phi$ despite representing a similar direction vector. Third, the angular error $\gamma$ directly measures the geodesic distance on the unit sphere between $\bm{\eta}$ and its estimate $\hat{\bm{\eta}}$, and is therefore invariant to rotations of the 3-D coordinate system. As a result, the MSAE provides a more physically meaningful and robust metric for direction vector estimation. Then, defining the physical angle vector as $\bm{\chi} \triangleq [\theta,\phi]^{\mathsf T}$, the lower-bound of $\mathbb{E}\{\gamma^2\}$ is given by \cite{nehorai1999performance}
	\begin{align}\label{CRB1-D}
		\mathbb{E}\{\gamma^2\} \geq {\rm MSAEB}_{\bm{\chi}}(\bm{R}) = {\rm CRB}_\theta(\bm{R}) + (\sin\theta)^2 {\rm CRB}_\phi(\bm{R}),
	\end{align}
	where ${\rm CRB}_\theta(\bm{R})$ and ${\rm CRB}_\phi(\bm{R})$ are the CRBs for estimating $\theta$ and $\phi$, respectively. Let $\bm{x} = [x_1, \dots, x_N]^{\mathsf T}$, $\bm{y} = [y_1, \dots, y_N]^{\mathsf T}$, and $\bm{z} = [z_1, \dots, z_N]^{\mathsf T}$ denote the antenna position vectors (APVs) along the $x$-, $y$-, and $z$-axes, respectively. The variance of the APV $\bm{x}$ is defined as ${\rm var}(\bm{x}) \triangleq \frac{1}{N}\sum_{n=1}^{N}x_n^2 - \mu(\bm{x})^2$, where $\mu(\bm{x}) = \frac{1}{N}\sum_{n=1}^{N} x_n$ denotes the mean of $\bm{x}$.
	Similarly, the covariance between two APVs $\bm{x}$ and $\bm{y}$ is defined as ${\rm cov}(\bm{x},\bm{y}) \triangleq \frac{1}{N} \sum_{n=1}^N x_n y_n - \mu(\bm{x}) \mu(\bm{y})$. Based on these definitions, we define the APV covariance matrix as
	\begin{align}\label{U}
		\bm{U} \triangleq \begin{bmatrix}
			{\rm var}(\bm{x}) & {\rm cov}(\bm{x},\bm{y}) & {\rm cov}(\bm{x},\bm{z})\\
			{\rm cov}(\bm{x},\bm{y}) & {\rm var}(\bm{y}) & {\rm cov}(\bm{y},\bm{z})\\
			{\rm cov}(\bm{x},\bm{z})& {\rm cov}(\bm{y},\bm{z}) & {\rm var}(\bm{z})
		\end{bmatrix}.
	\end{align}
	Moreover, we define the following vectors,
	\begin{align}
		\bm{f} &\triangleq \frac{\partial \bm{\eta}}{\partial \theta} =  [\cos\theta\cos\phi, \cos\theta\sin\phi, -\sin\theta]^{\mathsf T}, \\
		\bm{g} &\triangleq [-\sin\phi, \cos\phi, 0]^{\mathsf T},\notag
	\end{align}
	such that $\bm{g}\sin\theta = \frac{\partial \bm{\eta}}{\partial \phi}$. Then, the CRBs for physical angle estimation are given by the following theorem.
	\begin{theorem}
		${\rm CRB}_\theta(\bm{R})$ and ${\rm CRB}_\phi(\bm{R})$ are given by
		\begin{align}\label{CRB2}
			&{\rm CRB}_\theta(\bm{R}) \\
			&~~~~= \frac{\sigma^2\lambda^2}{8\pi^2PN|\beta|^2} \frac{\bm{g}^{\mathsf T} \bm{U} \bm{g}}{(\bm{f}^{\mathsf T} \bm{U} \bm{f})(\bm{g}^{\mathsf T} \bm{U} \bm{g})-(\bm{f}^{\mathsf T} \bm{U} \bm{g})^2}, \notag\\
			&{\rm CRB}_\phi(\bm{R}) \notag\\
			&~~~~= \frac{\sigma^2\lambda^2}{8\pi^2PN|\beta|^2}\frac{1}{(\sin\theta)^2} \frac{\bm{f}^{\mathsf T} \bm{U} \bm{f}}{(\bm{f}^{\mathsf T} \bm{U} \bm{f})(\bm{g}^{\mathsf T} \bm{U} \bm{g})-(\bm{f}^{\mathsf T} \bm{U} \bm{g})^2}.\notag
		\end{align}
	\end{theorem}
	\begin{proof}
		See Appendix A.
	\end{proof}
	Then, substituting \eqref{CRB2} back into \eqref{CRB1-D}, ${\rm MSAEB}_{\bm{\chi}}(\bm{R})$ can be further simplified as	
	\begin{align}\label{MSAEB}
		{\rm MSAEB}_{\bm{\chi}}(\bm{R}) &= \frac{\sigma^2\lambda^2}{8\pi^2PN|\beta|^2} \frac{\bm{g}^{\mathsf T} \bm{U} \bm{g} + \bm{f}^{\mathsf T} \bm{U} \bm{f}}{(\bm{f}^{\mathsf T} \bm{U} \bm{f})(\bm{g}^{\mathsf T} \bm{U} \bm{g})-(\bm{f}^{\mathsf T} \bm{U} \bm{g})^2}, \notag\\
		&\triangleq \rho F_{\bm{\chi}}(\bm{R}),
	\end{align}	
	where $\rho \triangleq \frac{\sigma^2\lambda^2}{8\pi^2PN|\beta|^2}$ and $F_{\bm{\chi}}(\bm{R}) \triangleq \frac{\bm{g}^{\mathsf T} \bm{U} \bm{g} + \bm{f}^{\mathsf T} \bm{U} \bm{f}}{(\bm{f}^{\mathsf T} \bm{U} \bm{f})(\bm{g}^{\mathsf T} \bm{U} \bm{g})-(\bm{f}^{\mathsf T} \bm{U} \bm{g})^2}$. Equation \eqref{MSAEB} indicates that ${\rm MSAEB}_{\bm{\chi}}(\bm{R})$ explicitly depends on the MA trajectory $\bm{R}$. Therefore, the MA trajectory can be optimized to minimize ${\rm MSAEB}_{\bm{\chi}}(\bm{R})$. Intuitively, this can be achieved by enlarging the spatial spreads of $\bm{x}$, $\bm{y}$, and $\bm{z}$ while maintaining symmetry of the trajectory with respect to (w.r.t.) the $x$-, $y$-, and $z$-axes, which helps increase $\bm{g}^{\mathsf T} \bm{U} \bm{g}$ and $\bm{f}^{\mathsf T} \bm{U} \bm{f}$ while reducing $\bm{f}^{\mathsf T}\bm{U}\bm{g}$. Nevertheless, a trade-off generally exists in minimizing ${\rm MSAEB}_{\bm{\chi}}(\bm{R})$ over different directions due to the coupling between $\bm{U}$ and $\{\bm{f},\bm{g}\}$, which depend on the MA trajectory $\bm{R}$ and the direction vector $\bm{\eta}$, respectively.
	
	Let $\mathbb{D}$ denote the angular region of interest for sensing the physical angle vector $\bm{\chi}$. To ensure robust sensing performance over directions in $\mathbb{D}$, we aim to minimize the maximum (worst-case) value of ${\rm MSAEB}_{\bm{\chi}}(\bm{R})$ over $\mathbb{D}$ by jointly optimizing the APM $\bm{R}$ and the AVM $\bm{V}$, which can be formulated as
	\begin{align}\label{MSAEBContinue}
		\min_{\bm{R},\bm{V}}\max_{\bm{\chi}\in\mathbb{D}} ~ {\rm MSAEB}_{\bm{\chi}}(\bm{R}).
	\end{align}
	Since obtaining the maximum value of ${\rm MSAEB}_{\bm{\chi}}(\bm{R})$ over $\mathbb{D}$ in closed-form is intractable, we uniformly discretize $\mathbb{D}$ into $Q$ grid points, denoted by $\bar{\mathbb{D}}\triangleq\{\bm{\chi}_q=[\theta_q,\phi_q]^{\mathsf T}\}_{q=1}^Q$. Accordingly, the objective in \eqref{MSAEBContinue} can be approximated as
	\begin{align}\label{MSAEBDiscrete}
		\min_{\bm{R},\bm{V}}\max_{1\leq q\leq Q} ~ {\rm MSAEB}_{\bm{\chi}_q}(\bm{R}).
	\end{align}
	Moreover, according to \eqref{MSAEB}, minimizing ${\rm MSAEB}_{\bm{\chi}}(\bm{R})$ w.r.t. $\bm{R}$ and $\bm{V}$ can be equivalently expressed as
	\begin{align}
		\min_{\bm{R},\bm{V}}  ~ {\rm MSAEB}_{\bm{\chi}}(\bm{R}) \iff 
		 \min_{\bm{R},\bm{V}}  ~ F_{\bm{\chi}}(\bm{R}).
	\end{align}
	As a result, the MA trajectory optimization problem can be formulated as
	\begin{subequations}
		\begin{align}
			\textrm {(P1)}~~\min_{\bm{R},\bm{V},\delta} \quad & \delta \label{P1a}\\
			\text{s.t.} \quad & F_{\bm{\chi}_q}(\bm{R}) \leq \delta, \quad q=1, \dots, Q, \label{P1b}\\
			& \|\bm{v}_n\|_2 \le v^{\rm m}, \quad n=1, \dots, N-1, \label{P1c} \\
			& \bm{r}_n \in \mathcal{C}, \quad n=1, \dots, N, \label{P1-D} \\
			& \bm{r}_n = \bm{r}_1 + T_s \sum_{m=1}^{n-1} \bm{v}_m, \quad n=2, \dots, N, \label{P1e}
		\end{align}
	\end{subequations}
	where $F_{\bm{\chi}_q}(\bm{R})\triangleq\frac{\bm{g}_q^{\mathsf T} \bm{U} \bm{g}_q + \bm{f}_q^{\mathsf T} \bm{U} \bm{f}_q}{(\bm{f}_q^{\mathsf T} \bm{U} \bm{f}_q)(\bm{g}_q^{\mathsf T} \bm{U} \bm{g}_q)-(\bm{f}_q^{\mathsf T} \bm{U} \bm{g}_q)^2}$, $\bm{f}_q \triangleq [\cos\theta_q\cos\phi_q, \cos\theta_q\sin\phi_q, -\sin\theta_q]^{\mathsf T}$, and $\bm{g}_q \triangleq [-\sin\phi_q, \cos\phi_q, 0]^{\mathsf T}$. Since the fractional constraint in \eqref{P1b} is non-convex w.r.t. $\bm{R}$ and $\bm{V}$, problem (P1) is inherently non-convex and challenging to solve optimally. In addition, the intrinsic coupling among $\bm{x}$, $\bm{y}$, and $\bm{z}$ further increases the complexity of solving this optimization problem.

	\section{Performance Analysis}
	Prior to solving problem (P1), in this section, we analyze the performance of the proposed MA-enabled sensing system to provide design insights for three specific cases: 1) $\mathbb{D}$ contains a single direction only. 2) 3-D isotropic sensing, i.e., ${\rm MSAEB}_{\bm{\chi}}(\bm{R})$ is constant for all $\bm{\chi}\in[0,\pi]\times[0,2\pi]$. 3) The MA trajectory lies entirely in the $x$-$y$ plane.
	
	\subsection{$\mathbb{D}$ Contains A Single Direction Only}
	
	\begin{figure}[!t]
		\centering
		\includegraphics[width=60mm]{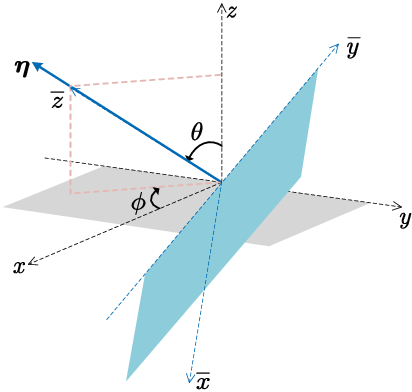}
		\caption{Illustration of rotated Cartesian coordinate system $\bar{x}$-$\bar{y}$-$\bar{z}$, where $\bar{z}$-axis is aligned with $\bm{\eta}$.}
		\label{new_coordinate}
	\end{figure}
	
	We first consider a special case where the angular region of interest $\mathbb{D}$ is only a single direction $\bm{\chi}$. In this case, the min-max objective function in \eqref{MSAEBContinue} reduces to minimizing ${\rm MSAEB}_{\bm{\chi}}(\bm{R})$ for one direction $\bm{\chi}$ only, i.e.,
	\begin{align}\label{MSAEBContinue1}
		\min_{\bm{R},\bm{V}} ~ {\rm MSAEB}_{\bm{\chi}}(\bm{R}).
	\end{align}
	To gain more insights, we exploit the rotational invariance of the MSAE metric and introduce a rotated Cartesian coordinate system, denoted by $\bar{x}$-$\bar{y}$-$\bar{z}$, such that the $\bar{z}$-axis is aligned with the direction vector $\bm{\eta}$ as shown in Fig.~\ref{new_coordinate}. In the new coordinate system, the elevation AoA is $\tilde{\theta}=0$, thus the target direction is
	\begin{align}
		\tilde{\bm{\eta}} &\triangleq[\sin\tilde{\theta} \cos\tilde{\phi}, \sin\tilde{\theta} \sin\tilde{\phi}, \cos\tilde{\theta}]^{\mathsf T} \\
		&= [0, 0, 1]^{\mathsf T}.\notag
	\end{align}
	Since ${\rm MSAEB}_{\bm{\chi}}(\bm{R})$ measures the geodesic distance on the unit sphere, it is invariant to rotations of the coordinate system, and therefore this transformation does not alter the value of the lower-bound of MSAE. Let $\bar{\bm{x}}=[\bar{x}_1,\ldots,\bar{x}_N]^{\mathsf T}$, $\bar{\bm{y}}=[\bar{y}_1,\ldots,\bar{y}_N]^{\mathsf T}$, and $\bar{\bm{z}}=[\bar{z}_1,\ldots,\bar{z}_N]^{\mathsf T}$ denote the APVs along the $\bar{x}$-, $\bar{y}$-, and $\bar{z}$-axes, respectively. Thus, the MA's position in the rotated coordinate system at time $t_n$ is
	\begin{align}
		\bar{\bm{r}}_n \triangleq [\bar{x}_n, \bar{y}_n, \bar{z}_n]^{\mathsf T} = \bm{Q}\bm{r}_n,
	\end{align}
	where $\bm{Q}\in\mathbb{R}^{3\times3}$ is an orthonormal rotation matrix given by
	\begin{align}
		\bm{Q} =
		\begin{bmatrix}
			\cos\theta\cos\phi & \cos\theta\sin\phi & -\sin\theta \\
			-\sin\phi & \cos\phi & 0 \\
			\sin\theta\cos\phi & \sin\theta\sin\phi & \cos\theta
		\end{bmatrix},
	\end{align}
	satisfying $\bm{Q}\bm{\eta} = \tilde{\bm{\eta}}$. Accordingly, the APV covariance matrix $\bm{U}$ in \eqref{U} can be expressed in the $\bar{x}$-$\bar{y}$-$\bar{z}$ coordinate system as
	\begin{align}\label{U2}
		\bar{\bm{U}} =
		\begin{bmatrix}
			{\rm var}(\bar{\bm{x}}) & {\rm cov}(\bar{\bm{x}},\bar{\bm{y}}) & {\rm cov}(\bar{\bm{x}},\bar{\bm{z}}) \\
			{\rm cov}(\bar{\bm{x}},\bar{\bm{y}}) & {\rm var}(\bar{\bm{y}}) & {\rm cov}(\bar{\bm{y}},\bar{\bm{z}}) \\
			{\rm cov}(\bar{\bm{x}},\bar{\bm{z}}) & {\rm cov}(\bar{\bm{y}},\bar{\bm{z}}) & {\rm var}(\bar{\bm{z}})
		\end{bmatrix}.
	\end{align}
	Since $\tilde{\theta}=0$, $\bm{f}$ and $\bm{g}$ in the $\bar{x}$-$\bar{y}$-$\bar{z}$ coordinate system are given by
	\begin{align}\label{fg2}
		\bar{\bm{f}} &= [\cos\tilde{\theta}\cos\tilde{\phi}, \cos\tilde{\theta}\sin\tilde{\phi}, -\sin\tilde{\theta}]^{\mathsf T}\\
		&=[\cos\tilde{\phi}, \sin\tilde{\phi}, 0]^{\mathsf T}, \notag\\
		\bar{\bm{g}} &= [-\sin\tilde{\phi}, \cos\tilde{\phi}, 0]^{\mathsf T}. \notag
	\end{align}
	Substituting \eqref{U2} and \eqref{fg2} into \eqref{MSAEB}, we obtain
	\begin{align}\label{fUf}
		\bar{\bm{f}}^{\mathsf T}\bar{\bm{U}}\bar{\bm{f}} &= (\cos\tilde{\phi})^2{\rm var}(\bar{\bm{x}}) + (\sin\tilde{\phi})^2{\rm var}(\bar{\bm{y}}) \\
		&~~~~ + 2\sin\tilde{\phi}\cos\tilde{\phi}{\rm cov}(\bar{\bm{x}},\bar{\bm{y}}), \notag\\
		\bar{\bm{g}}^{\mathsf T}\bar{\bm{U}}\bar{\bm{g}} &= (\sin\tilde{\phi})^2{\rm var}(\bar{\bm{x}}) + (\cos\tilde{\phi})^2{\rm var}(\bar{\bm{y}}) \notag\\
		&~~~~ - 2\sin\tilde{\phi}\cos\tilde{\phi}{\rm cov}(\bar{\bm{x}},\bar{\bm{y}}), \notag\\
		\bar{\bm{f}}^{\mathsf T}\bar{\bm{U}}\bar{\bm{g}} &= \sin\tilde{\phi}\cos\tilde{\phi}({\rm var}(\bar{\bm{y}}) - {\rm var}(\bar{\bm{x}})) \notag\\
		&~~~~ + ((\cos\tilde{\phi})^2 - (\sin\tilde{\phi})^2){\rm cov}(\bar{\bm{x}},\bar{\bm{y}}). \notag
	\end{align}
	It is worth noting that all terms associated with $\bar{\bm{z}}$ vanish in \eqref{fUf}, indicating that antenna movement along the target direction does not contribute to improving the direction estimation accuracy. Thus, $F_{\bm{\chi}}(\bm{R})$ for a single direction can be simplified as
	\begin{align}\label{F2D}
		F_{\bm{\chi}}(\bm{R}) &= \frac{\bar{\bm{g}}^{\mathsf T} \bar{\bm{U}} \bar{\bm{g}} + \bar{\bm{f}}^{\mathsf T} \bar{\bm{U}} \bar{\bm{f}}}{(\bar{\bm{f}}^{\mathsf T} \bar{\bm{U}} \bar{\bm{f}})(\bar{\bm{g}}^{\mathsf T} \bar{\bm{U}} \bar{\bm{g}})-(\bar{\bm{f}}^{\mathsf T} \bar{\bm{U}} \bar{\bm{g}})^2} \\
		&= \frac{{\rm var}(\bar{\bm{x}})+{\rm var}(\bar{\bm{y}})}
		{{\rm var}(\bar{\bm{x}}){\rm var}(\bar{\bm{y}})
			-{\rm cov}(\bar{\bm{x}},\bar{\bm{y}})^2}. \notag
	\end{align}		
	Equation \eqref{F2D} indicates that, to reduce the lower-bound of MSAE, the spatial spreads of the MA trajectory in the $\bar{x}$-$\bar{y}$ plane should be enlarged to increase ${\rm var}(\bar{\bm{x}})$ and ${\rm var}(\bar{\bm{y}})$. Meanwhile, the MA trajectory should be symmetric w.r.t. the $\bar{x}$- and $\bar{y}$-axes so as to reduce ${\rm cov}(\bar{\bm{x}}, \bar{\bm{y}})^2$. In contrast, since $\bar{\bm{z}}$ does not appear in $F_{\bm{\chi}}(\bm{R})$, any antenna movement along the $\bar{z}$-axis is ineffective. Therefore, when $\mathbb{D}$ contains only a single direction, the optimal MA trajectory should lie entirely in the plane orthogonal to the target direction, i.e.,
	\begin{align}
		\bar{z}_n = \text{constant}, \quad n=1,\ldots,N.
	\end{align}

	\subsection{3-D Isotropic Sensing}
	
	In this subsection, we investigate the case of 3-D isotropic sensing, i.e., ${\rm MSAEB}_{\bm{\chi}}(\bm{R})$ remains invariant w.r.t. the target direction $\bm{\chi}$ over the entire angular region $[0,\pi]\times[0,2\pi]$. Such isotropy is highly desirable in practice, since it guarantees uniform sensing performance for the target direction.
	
	\begin{theorem}\label{theorem2}
		${\rm MSAEB}_{\bm{\chi}}(\bm{R})$ is constant for all
		$\bm{\chi}\in[0,\pi]\times[0,2\pi]$ if and only if
		\begin{align}
			\bm{U} = \tau\bm{I}_3,
		\end{align}
		where $\tau = {\rm var}(\bm{x}) = {\rm var}(\bm{y}) = {\rm var}(\bm{z})$.
	\end{theorem}
	\begin{proof}
		See Appendix~B.
	\end{proof}
	Theorem~\ref{theorem2} indicates that to achieve isotropic wireless sensing performance, the variances of the MA trajectory along the three Cartesian axes are identical and the cross-axis	correlations of the MA trajectory vanish. Under this condition, the trajectory does not favor any particular spatial direction, leading to direction-independent sensing	accuracy. Moreover, when $\bm{U}=\tau\bm{I}_3$, the lower-bound of MSAE can be simplified as
	\begin{align}
		{\rm MSAEB}_{\bm{\chi}}(\bm{R})
		= \frac{2\rho}{\tau}. \notag
	\end{align}
	
	\begin{figure}[!t]
		\centering
		\includegraphics[width=85mm]{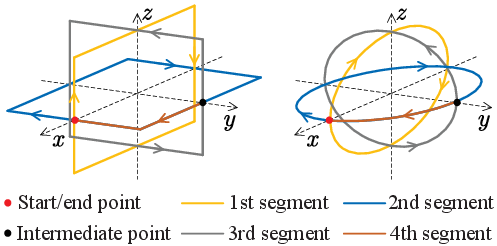}
		\caption{Illustration of the isotropic MA trajectory formed by three mutually orthogonal squares (left) and circles (right). The movement sequence follows the order of orange, blue, gray, and brown segments.}
		\label{isotropic_trajectory}
	\end{figure}
		
	A representative example of such an isotropic MA trajectory is formed by three mutually orthogonal regular $4k$-gons, with $k\in\mathbb{Z}^+$, lying in the $x$-$y$, $y$-$z$, and $x$-$z$ planes, respectively. As $k\to\infty$, these polygons converge to three mutually orthogonal circular trajectories. As illustrated in Fig.~\ref{isotropic_trajectory}, the MA starts from the red point on the $x$-axis, completes one loop along the orange square/circular segment in the $x$-$z$ plane, follows the blue square/circular segment in the $x$-$y$ plane for $3/4$ of a loop to reach the black point on the $y$-axis, completes one loop along the gray square/circular segment in the $y$-$z$ plane, and finally follows the brown square/circular segment in the $x$-$y$ plane for $1/4$ of a loop to return to the starting red point.
	
	\subsection{MA Trajectory in $x$-$y$ Plane}
	In this subsection, we consider the case where the MA trajectory lies entirely in the $x$-$y$ plane, which implies $z_n = 0$ for all $n=1,\ldots,N$. Consequently, the APV along the $z$-axis is a zero vector, leading to ${\rm var}(\bm{z}) = 0$ and ${\rm cov}(\bm{x},\bm{z}) = {\rm cov}(\bm{y},\bm{z}) = 0$. In this case, the sensing performance exhibits a strong dependence on the target's elevation angle $\theta$, which is shown in the following theorem.
	
	\begin{theorem}\label{theorem_xy_plane}
		For any MA trajectory in the $x$-$y$ plane, and for a given azimuth angle $\phi \in [0, 2\pi]$, the function $F_{\bm{\chi}}(\bm{R})$ monotonically increases w.r.t. the elevation angle $\theta \in [0, \pi/2)$, and monotonically decreases w.r.t. $\theta \in (\pi/2, \pi]$.
	\end{theorem}
	\begin{proof}
		See Appendix~C.
	\end{proof}
	
	Theorem~\ref{theorem_xy_plane} reveals that a 2-D planar MA trajectory yields the highest direction estimation accuracy when the target is located near the poles (i.e., $\theta \approx 0$ or $\theta \approx \pi$). However, as the target approaches the endfire direction of the $x$-$y$ plane (i.e., $\theta \to \pi/2$), the estimation error diverges. This is because the MA trajectory in the $x$-$y$ plane yields a small spatial aperture over the in-plane-direction orthogonal to $z$-axis, rendering it incapable of resolving the elevation angle of a target located in the $x$-$y$ plane. Notably, Theorem~\ref{theorem_xy_plane} is not restricted to the $x$-$y$ plane and can be directly generalized to any 2-D planar trajectory in 3-D space. This finding highlights the necessity of 3-D MA movement to ensure reliable target sensing over 3-D space in general.

	\section{Optimization Algorithm}
	In this section, we first present an algorithm to solve problem (P1). Then, we develop a simplified algorithm for solving problem (P3) in the special case where $\mathbb{D}$ contains a single direction only.
	
	\subsection{Algorithm for Solving Problem (P1)}
	We assume that $\mathcal{C}$ is a convex 3-D region, which ensures the convexity of the constraint in \eqref{P1-D}{\footnote{If the movement region $\mathcal{C}$ is non-convex, a maximal inscribed convex sub-region can be extracted via iterative regional inflation \cite{deits2015computing}. The 3-D MA trajectory can then be optimized within this convex sub-region using Algorithm~\ref{alg1}.}}. Although the constraints \eqref{P1c}, \eqref{P1-D}, and \eqref{P1e} in problem (P1) are convex w.r.t. $\bm{R}$ and $\bm{V}$, the function $F_{\bm{\chi}_q}(\bm{R})$ in \eqref{P1b} contains a non-convex fractional form, which makes problem (P1) difficult to solve directly. To address this issue, we adopt the SCA approach.	To facilitate convex reformulation, we first express the variance function ${\rm{var}}(\bm{x})$ and covariance function ${\rm{cov}}(\bm{x},\bm{y})$ in quadratic form:
	\begin{align}\label{B}
		{\rm{var}}(\bm{x}) &\triangleq \bm{x}^{\mathsf T} \bm{B} \bm{x}, \notag\\
		{\rm{cov}}(\bm{x},\bm{y}) &\triangleq \bm{x}^{\mathsf T} \bm{B} \bm{y},
	\end{align}
	where $\bm{B}\triangleq \frac{1}{N}\bm{I}_N-\frac{1}{N^2}\bm{1}_N\bm{1}_N^{\mathsf T}$ is a positive semi-definite (PSD) matrix, and $\bm{1}_N$ denotes the $N$-dimensional all-ones vector. Then, the APV covariance matrix $\bm{U}$ in \eqref{U} can be equivalently expressed as
	\begin{equation}
		\bm{U} = \bm{R} \bm{B} \bm{R}^{\mathsf T}.
	\end{equation}
	Moreover, by defining $\bm{\Phi} \triangleq [\bm{f}, \bm{g}] \in \mathbb{R}^{3 \times 2}$, we have $\bm{\Phi}^{\mathsf T} \bm{U} \bm{\Phi} = \begin{bmatrix}
		\bm{f}^{\mathsf T} \bm{U} \bm{f} & \bm{f}^{\mathsf T} \bm{U} \bm{g} \\
		\bm{f}^{\mathsf T} \bm{U} \bm{g} & \bm{g}^{\mathsf T} \bm{U} \bm{g}
	\end{bmatrix}$. Then, using the inversion formula for a $2\times2$ matrix, i.e., $\begin{bmatrix}
		a & b \\
		c & d
	\end{bmatrix}^{-1}=\frac{1}{ad-bc}\begin{bmatrix}
		d & -b \\
		-c & a
	\end{bmatrix}$, $F_{\bm{\chi}}(\bm{R})$ can be written as the trace of a $2 \times 2$ inverse matrix:
	\begin{equation}
		F_{\bm{\chi}}(\bm{R}) = \text{Tr}\left( (\bm{\Phi}^{\mathsf T} \bm{U} \bm{\Phi})^{-1} \right).
	\end{equation}
	The function $\text{Tr}(\bm{A}^{-1})$ is convex and matrix-decreasing over the positive definite cone, i.e., $\text{Tr}(\bm{A}^{-1})\leq\text{Tr}(\bm{B}^{-1})$ when $\bm{A}\succeq\bm{B}$. To construct a convex upper-bound for $F_{\bm{\chi}}(\bm{R})$, we linearize the inner APV covariance matrix $\bm{U}$. Let $\bm{R}^p$ denote the APM obtained at the $p$th iteration of SCA. Since $\bm{U}$ is quadratic of $\bm{R}$, its first-order Taylor expansion at $\bm{R}^p$ yields the following global matrix lower-bound:
	\begin{align}
		\bm{U} &\succeq \bar{\bm{U}}(\bm{R}|\bm{R}^p) \\
		&\triangleq \bm{R}^p \bm{B} (\bm{R}^p)^{\mathsf T} + \bm{R}^p \bm{B} (\bm{R} - \bm{R}^p)^{\mathsf T} + (\bm{R} - \bm{R}^p) \bm{B} (\bm{R}^p)^{\mathsf T} \notag\\
		&= \bm{R}^p \bm{B} \bm{R}^{\mathsf T} + \bm{R} \bm{B} (\bm{R}^p)^{\mathsf T} - \bm{R}^p \bm{B} (\bm{R}^p)^{\mathsf T}. \notag
	\end{align}
	Due to the matrix-decreasing property of the function $\text{Tr}(\bm{A}^{-1})$, replacing $\bm{U}$ with its lower-bound $\bar{\bm{U}}(\bm{R}|\bm{R}^p)$ yields the following convex surrogate function:
	\begin{align}\label{surrogate}
		F_{\bm{\chi}}(\bm{R}) \leq \bar{F}_{\bm{\chi}}(\bm{R} | \bm{R}^p) \triangleq \text{Tr}\left( \left( \bm{\Phi}^{\mathsf T} \bar{\bm{U}}(\bm{R}|\bm{R}^p) \bm{\Phi} \right)^{-1} \right).
	\end{align}
	This surrogate function is convex because $\bar{\bm{U}}(\bm{R}|\bm{R}^p)$ is affine of $\bm{R}$ and $\text{Tr}(\bm{A}^{-1})$ is convex for $\bm{A}\succeq0$. Therefore, at the $p$th SCA iteration, the MA trajectory is updated by solving the following convex optimization problem:
	\begin{subequations}
		\begin{align}
			\textrm {(P2)}~~\min_{\bm{R},\bm{V},\delta} \quad & \delta \label{P3a}\\
			\text{s.t.} \quad & \bar{F}_{\bm{\chi}_q}(\bm{R} | \bm{R}^p) \leq \delta, \quad q=1, \dots, Q, \label{P3b} \\
			&\eqref{P1c},~\eqref{P1-D},~\eqref{P1e}. \notag
		\end{align}
	\end{subequations}
	Problem (P2) is a convex optimization problem since each constraint is convex w.r.t. $\bm{R}$, $\bm{V}$, and $\delta$. Therefore, it can be efficiently solved using off-the-shelf convex optimization solvers, e.g., CVX \cite{grantcvx}.
	
	Based on the solution to problem (P2), the overall algorithm for solving problem (P1) is summarized in Algorithm~\ref{alg1}. Specifically, in lines 3–7, the APM $\bm{R}$ and AVM $\bm{V}$ are iteratively refined by solving problem (P2). The iteration proceeds until the variation of $\delta$ in \eqref{P3a} between two successive updates becomes smaller than a prescribed tolerance $\epsilon$.
	
	\begin{algorithm}[!t]
		\caption{SCA Algorithm for Problem (P1)}
		\label{alg1}
		\begin{algorithmic}[1]
			\STATE \textbf{Input:} $N$, $\{\bm{\chi}_q\}_{q=1}^Q$, $\epsilon$, $\bm{R}^0$, $\mathcal{C}$.
			\STATE \textbf{Initialize:} $p \gets 0$, $\bm{R} \gets \bm{R}^0$, $\delta \gets 0$.
			
			\WHILE{The increment of $\delta$ is greater than $\epsilon$}
			\STATE Solve problem (P2) to obtain $\bm{R}^{p+1}$.
			\STATE $p \gets p + 1$.
			\STATE Update $\delta \gets \max_{1\leq q\leq Q} F_{\bm{\chi}_q}(\bm{R}^p)$.
			\ENDWHILE
			
			\STATE \textbf{Output:} $\bm{R}$, $\bm{V}$.
		\end{algorithmic}
	\end{algorithm}
	
	Since the objective value is non-increasing throughout the iterations and is lower-bounded by zero, Algorithm~\ref{alg1} is guaranteed to converge. The computational complexity of solving convex problem (P2) via the interior-point method is $\mathcal{O}\big((N(N+Q)^{0.5}(9N^2+N+Q))\ln(1/\kappa)\big)$, where $\kappa$ represents the solution accuracy of the interior-point method. Let $I$ denote the maximum number of SCA iterations (lines 3–7). Then, the total computational complexity of Algorithm~\ref{alg1} is given by $\mathcal{O}\big((N(N+Q)^{0.5}(9N^2+N+Q))\ln(1/\kappa)I\big)$.
	
	\subsection{Simplified Algorithm When $\mathbb{D}$ Contains A Single Direction Only}
	As discussed in Section III-A, when the angular region of interest $\mathbb{D}$ contains a single target direction only, the 3-D trajectory optimization problem reduces to a 2-D trajectory optimization in the $\bar{x}$-$\bar{y}$ plane orthogonal to the target direction. Consequently, the APV along the $\bar{z}$-axis can be omitted from the optimization variables, significantly reducing the problem dimension. Accordingly, the MA trajectory optimization problem can be simplified as
	\begin{subequations}
		\begin{align}
			\textrm{(P3)}~~
			\min_{\bm{R},\bm{V}} \quad
			& \frac{{\rm var}(\bar{\bm{x}})+{\rm var}(\bar{\bm{y}})}
			{{\rm var}(\bar{\bm{x}}){\rm var}(\bar{\bm{y}})
				-{\rm cov}(\bar{\bm{x}},\bar{\bm{y}})^2} \label{P2a} \\
			\text{s.t.} \quad
			& \bar{\bm{r}}_n = \bm{Q}\bm{r}_n, \quad n=1, \dots, N, \label{P2b} \\
			&\eqref{P1c},~\eqref{P1-D},~\eqref{P1e}. \notag
		\end{align}
	\end{subequations}
	
	Let $\bm{R}_{xy} \triangleq [\bar{\bm{x}}, \bar{\bm{y}}]^{\mathsf T} \in \mathbb{R}^{2 \times N}$ denote the rotated 2-D APM in the $\bar{x}$-$\bar{y}$ coordinate system. The $2 \times 2$ APV covariance matrix in the $\bar{x}$-$\bar{y}$ plane can be expressed as 
	\begin{align}\label{Uxy}
		\bm{U}_{xy} &\triangleq \begin{bmatrix}
			{\rm var}(\bar{\bm{x}}) & {\rm cov}(\bar{\bm{x}},\bar{\bm{y}}) \\
			{\rm cov}(\bar{\bm{x}},\bar{\bm{y}}) & {\rm var}(\bar{\bm{y}})
		\end{bmatrix} \\
		&= \bm{R}_{xy} \bm{B} \bm{R}_{xy}^{\mathsf T}. \notag
	\end{align}
	Then, the objective function in \eqref{P2a} can be represented as the trace of the inverse of $\bm{U}_{xy}$, i.e.,
	\begin{align}
		\frac{{\rm var}(\bar{\bm{x}})+{\rm var}(\bar{\bm{y}})}
		{{\rm var}(\bar{\bm{x}}){\rm var}(\bar{\bm{y}})
			-{\rm cov}(\bar{\bm{x}},\bar{\bm{y}})^2} = \text{Tr}\left( \bm{U}_{xy}^{-1} \right).
	\end{align}	
	Therefore, similar to problem (P1), we can apply the SCA approach to handle the non-convexity of the objective function $\text{Tr}\left( \bm{U}_{xy}^{-1} \right)$. Specifically, let $\bm{R}_{xy}^p$ denote the rotated 2-D APM obtained at the $p$th iteration of SCA. By taking the first-order Taylor expansion of $\bm{U}_{xy}$ at $\bm{R}_{xy}^p$, we obtain the following global matrix lower-bound:
	\begin{align}
		\bm{U}_{xy} &\succeq \bar{\bm{U}}_{xy}(\bm{R}_{xy}|\bm{R}_{xy}^p) \\ &\triangleq \bm{R}_{xy}^p \bm{B} \bm{R}_{xy}^{\mathsf T} + \bm{R}_{xy} \bm{B} (\bm{R}_{xy}^p)^{\mathsf T}  - \bm{R}_{xy}^p \bm{B} (\bm{R}_{xy}^p)^{\mathsf T}. \notag
	\end{align}
	Then, replacing $\bm{U}_{xy}$ with its matrix lower-bound $\bar{\bm{U}}_{xy}(\bm{R}_{xy}|\bm{R}_{xy}^p)$ yields the following convex surrogate function for problem (P3):
	\begin{align}
		\text{Tr}\left( \bm{U}_{xy}^{-1} \right) \leq \text{Tr}\left( \bar{\bm{U}}_{xy}(\bm{R}_{xy}|\bm{R}_{xy}^p)^{-1} \right).
	\end{align}
	Therefore, at the $p$th SCA iteration, the MA trajectory is updated by solving the following simplified convex optimization problem:
	\begin{subequations}
		\begin{align}
			\textrm{(P4)}~~\min_{\bm{R}, \bm{V}} \quad & \text{Tr}\left( \bar{\bm{U}}_{xy}(\bm{R}_{xy}|\bm{R}_{xy}^p)^{-1} \right) \label{P4a} \\
			\text{s.t.} \quad &\eqref{P2b},~\eqref{P1c},~\eqref{P1-D},~\eqref{P1e}. \notag
		\end{align}
	\end{subequations}
	Problem (P4) is a convex optimization problem and can be efficiently solved using off-the-shelf convex optimization solvers, e.g., CVX \cite{grantcvx}. Since $\mathbb{D}$ contains only a single direction, the computational complexity of solving (P4) using the interior-point method is reduced to $\mathcal{O}\big(N^{3.5}\ln(1/\kappa)\big)$, thereby significantly reducing the computational complexity in this special case.

	\section{Numerical Results}

	\begin{figure}[!t]
		\centering
		\subfigure[Without the region size constraint]{
			\begin{minipage}{.47\textwidth}
				\centering
				\includegraphics[scale=.48]{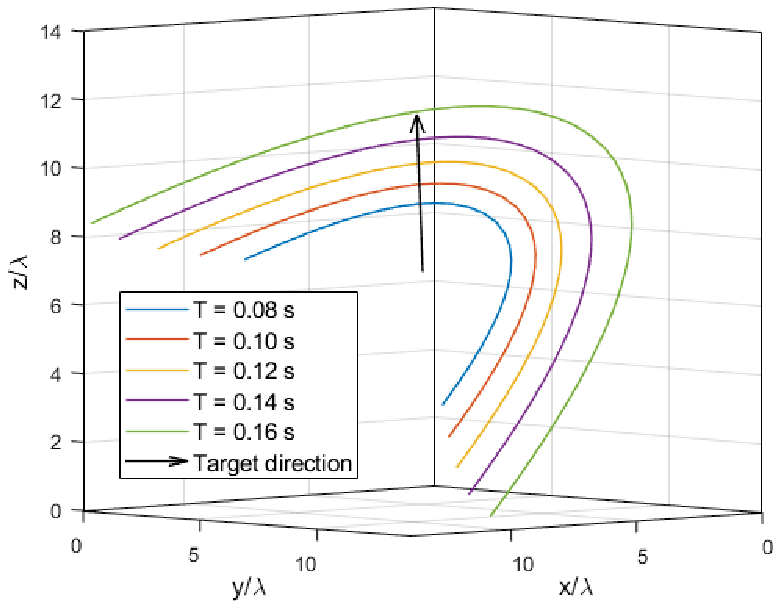}
			\end{minipage}
			\label{Tra_opt_case1_infinite}
		}
		\subfigure[With the region size constraint]{
			\begin{minipage}{.47\textwidth}
				\centering
				\includegraphics[scale=.48]{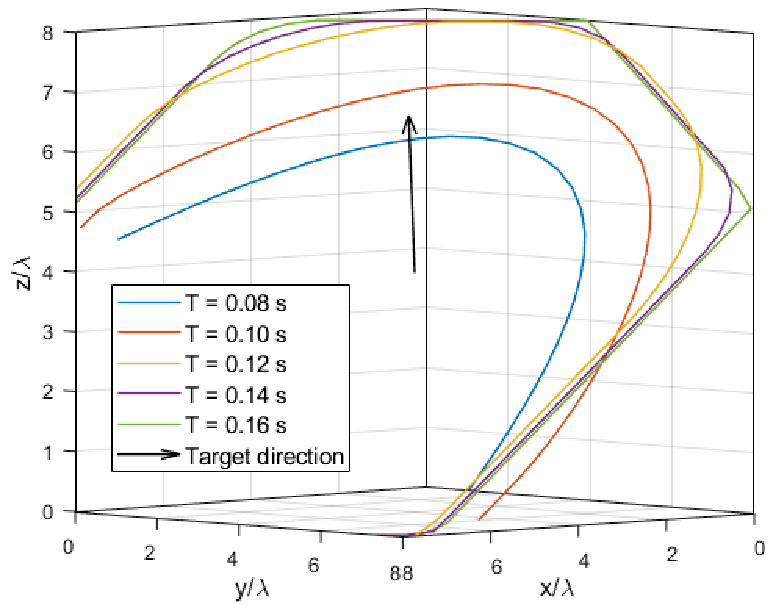}
			\end{minipage}
			\label{Tra_opt_case1_finite}
		}
		\caption{Illustration of MA trajectories when $\mathbb{D}$ contains a single direction.}
		\label{FIG5}
	\end{figure}
	
	\begin{figure}[!t]
		\centering
		\includegraphics[width=70mm]{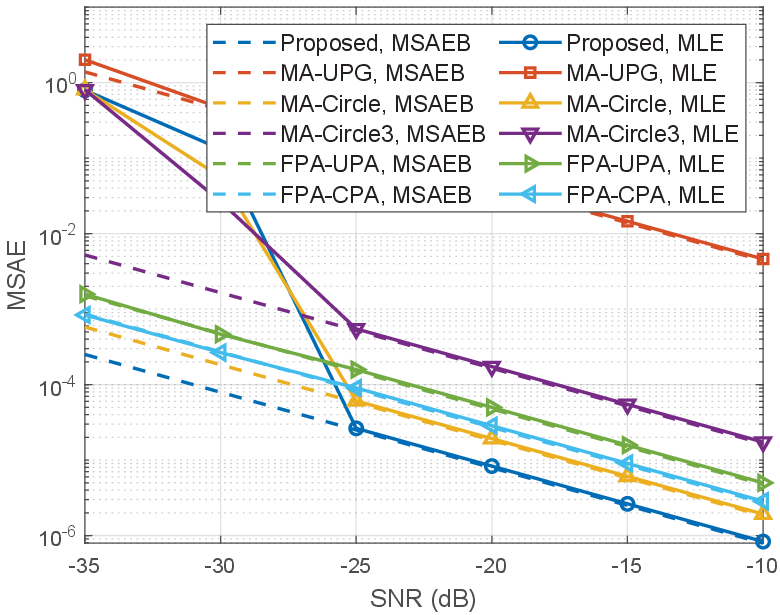}
		\caption{MSAE versus SNR for different MA trajectories when $\mathbb{D}$ contains a single direction.}
		\label{FIG6}
	\end{figure}
	
	\begin{figure}[!t]
		\centering
		\subfigure[Proposed]{
			\begin{minipage}{.14\textwidth}
				\centering
				\includegraphics[scale=.37]{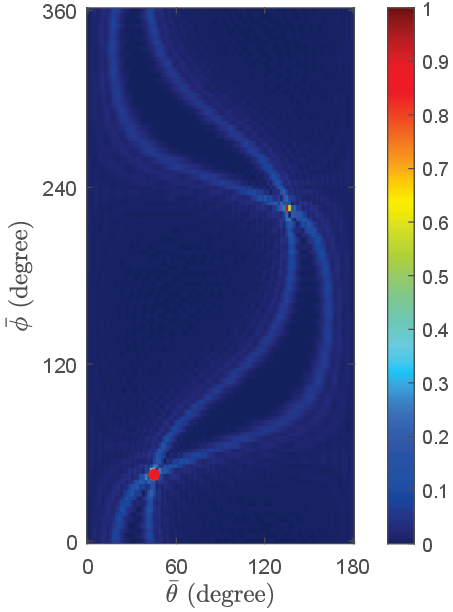}
			\end{minipage}
			\label{correlation_proposed}
		}
		\subfigure[MA-UPG]{
			\begin{minipage}{.14\textwidth}
				\centering
				\includegraphics[scale=.37]{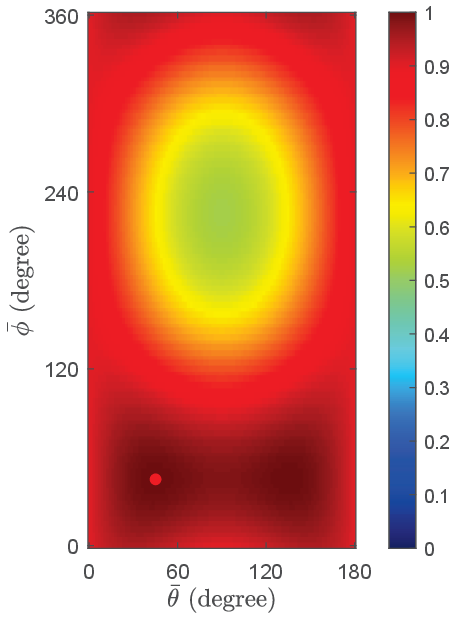}
			\end{minipage}
			\label{correlation_upa}
		}
		\subfigure[MA-Circle]{
			\begin{minipage}{.14\textwidth}
				\centering
				\includegraphics[scale=.37]{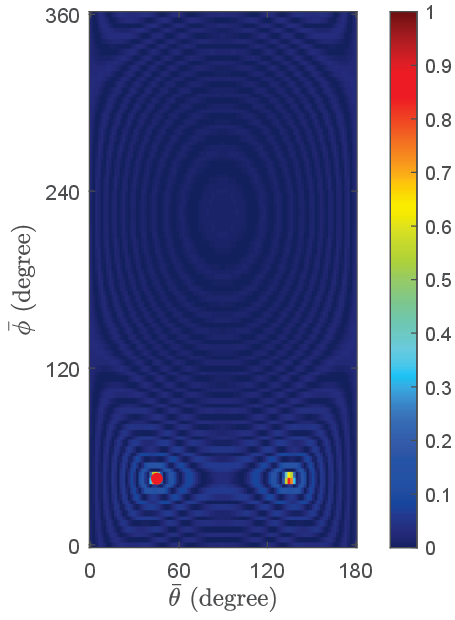}
			\end{minipage}
			\label{correlation_circle}
		}
		\\
		\subfigure[MA-Circle3]{
			\begin{minipage}{.14\textwidth}
				\centering
				\includegraphics[scale=.37]{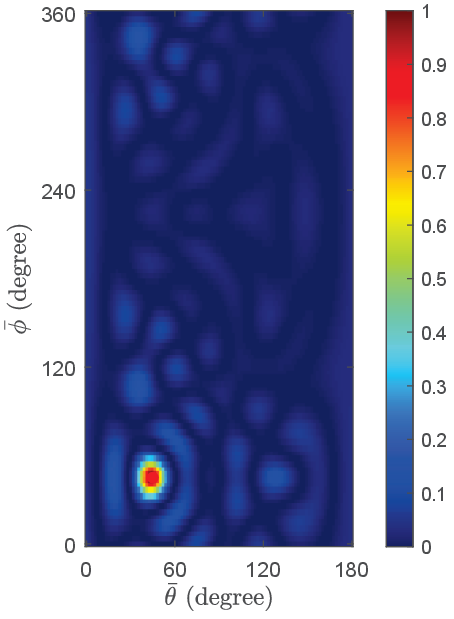}
			\end{minipage}
			\label{correlation_circle3}
		}
		\subfigure[FPA-UPA]{
			\begin{minipage}{.14\textwidth}
				\centering
				\includegraphics[scale=.37]{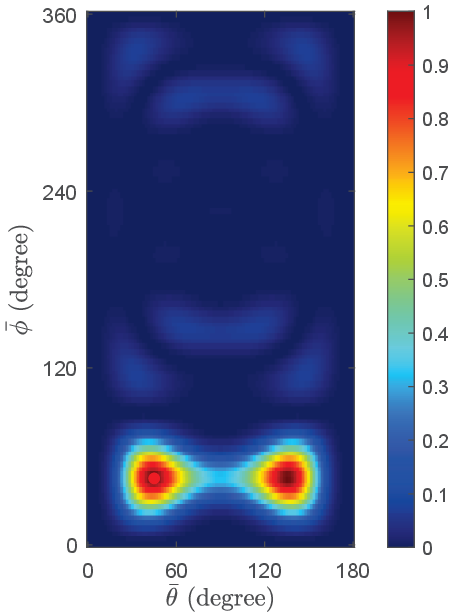}
			\end{minipage}
			\label{correlation_FPA_UPA}
		}
		\subfigure[FPA-CPA]{
			\begin{minipage}{.14\textwidth}
				\centering
				\includegraphics[scale=.37]{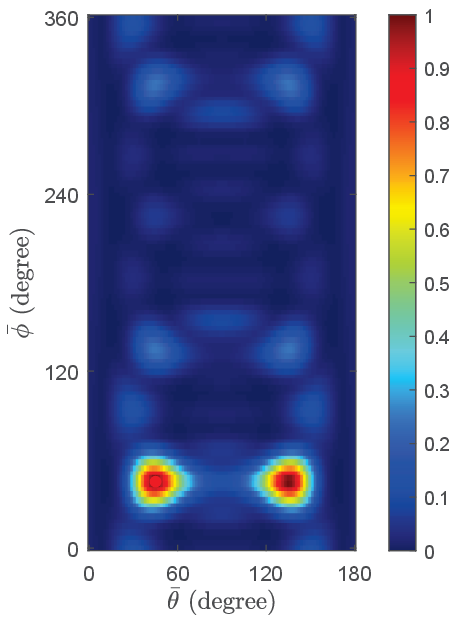}
			\end{minipage}
			\label{correlation_FPA_CPA}
		}
		\caption{Comparison of steering vector correlation with different MA trajectories when $\mathbb{D}$ contains a single direction.}
		\label{FIG7}
	\end{figure}

In this section, we provide numerical results to evaluate the performance of the proposed 3-D MA trajectory optimization for target direction vector estimation. The system parameters are configured as follows: the sampling period is $T_s=10$~µs, the carrier wavelength is $\lambda = 0.05$~m corresponding to the carrier frequency of $6$~GHz, and the maximum speed of the MA is $v^{\rm m} = 10$~m/s. The maximum moving distance of the MA between two adjacent snapshots is thus $\Delta \triangleq v^{\rm m} T_s = 2 \times 10^{-3} \lambda$. The average received SNR is defined as $P|\beta|^2/\sigma^2$. The antenna movement region $\mathcal{C}$ is a cubic region of size $A \times A \times A$. For the SCA-based trajectory optimization in Algorithm~\ref{alg1}, the convergence threshold is set as $\epsilon = 10^{-4}$. To facilitate the optimization, every $250$ consecutive elements in the AVM $\bm{V}$ share an identical value.

To comprehensively evaluate the performance of the proposed MA trajectory, we consider three benchmark MA trajectories for comparison:
1) \textbf{Uniform planar grids (UPG)}: The MA moves at the maximum speed $v^{\rm m}$ to visit the points of $\sqrt{N}\times \sqrt{N}$ UPG in the $x$-$y$ plane, with adjacent grid spacing being $\Delta$.
2) \textbf{Circle}: The MA moves at $v^{\rm m}$ to form a circular trajectory with radius $\frac{\Delta}{2\sin\left(\tfrac{\pi}{N}\right)}$ in the $x$-$y$ plane.
3) \textbf{Circle3}: As shown in Fig.~\ref{isotropic_trajectory}, the MA moves at $v^{\rm m}$ to form an isotropic 3-D trajectory consisting of three orthogonal circles of equal radius $\frac{\Delta}{2\sin\left(\tfrac{3\pi}{N}\right)}$ located in the $x$-$y$, $x$-$z$, and $y$-$z$ planes, respectively. Moreover, we consider the FPA receiver equipped with an array of $16$ antennas to estimate the target direction vector over $N$ snapshots. The benchmark FPA array schemes considered are as follows: 1) \textbf{Uniform planar array (UPA)}: The $4\times 4$ UPA with half-wavelength inter-antenna spacing.
2) \textbf{Coprime planar array (CPA)}: The $4\times 4$ CPA with the 2D APM $\bm{R}_{xy}^{\rm CPA}\in\mathbb{R}^{2 \times 16}$ given by $\bm{R}_{xy}^{\rm CPA}[p,q] \in \frac{\lambda}{2}\{0,2,3,4\}$, for $1\leq p\leq 2$ and $1\leq q\leq 16$ \cite{zheng2017generalized}.

\subsection{$\mathbb{D}$ Contains A Single Direction Only}

First, we investigate the special case where the angular region $\mathbb{D}$ contains only a single target direction. The physical elevation and azimuth AoAs are set as $\theta = 45^\circ$ and $\phi = 45^\circ$. We apply the simplified optimization algorithm for problem (P3) to obtain the proposed MA trajectory.

Fig.~\ref{FIG5} illustrates the optimized MA trajectories under different total sensing time $T$. In Fig.~\ref{Tra_opt_case1_infinite}, the antenna movement region is assumed to be sufficiently large such that constraint \eqref{P1-D} is inactive. In contrast, Fig.~\ref{Tra_opt_case1_finite} considers a finite movement region with size $A = 8\lambda$. It is observed that the optimized MA trajectory is strictly confined in the 2-D plane orthogonal to the target direction $\bm{\eta}$, which fully corroborates the theoretical result in Section III-A that antenna movement along the target direction yields no estimation gain. Moreover, as $T$ increases, the MA trajectory progressively expands within the orthogonal plane, effectively enlarging the virtual aperture formed by antenna movement. This expansion improves the angular resolution and consequently enhances the direction vector estimation accuracy.

Fig.~\ref{FIG6} compares the direction vector estimation MSAEs versus SNR for different MA trajectories and FPA arrays. We set $T = 0.16$~s and $A=15\lambda$. The MLE of the direction vector is obtained via \eqref{MLE1-D}. It is observed that the MLE curves tightly match the theoretical MSAEB for all schemes when the SNR is sufficiently large (i.e., SNR $\geq -25$~dB). Moreover, the proposed trajectory significantly outperforms all benchmark schemes in this regime. This is because the benchmark trajectories and FPA arrays are constrained in the horizontal $x$-$y$ plane, their effective apertures projected onto the plane orthogonal to $\bm{\eta}$ are severely degraded, leading to poor estimation accuracy. Furthermore, although FPA arrays provide array gain compared with the proposed single MA-enabled  sensing scheme, antenna movement can synthesize a significantly larger virtual aperture, thereby achieving higher angular resolution.

To gain more insights, Fig.~\ref{FIG7} shows the normalized steering vector correlation $\frac{1}{N^2}|\bm{\alpha}(\bm{R},\bm{\eta})^{\mathsf H} \bm{\alpha}(\bm{R},\bar{\bm{\eta}})|^2$, where $\bar{\bm{\eta}} = [\sin\bar{\theta} \cos\bar{\phi}, \sin\bar{\theta} \sin\bar{\phi}, \cos\bar{\theta}]^{\mathsf T}$ with $\bar{\theta}\in[0,\pi]$ and $\bar{\phi}\in[0,2\pi]$, and $\bm{\eta}=[\sin\theta \cos\phi, \sin\theta \sin\phi, \cos\theta]^{\mathsf T}=[0.5,0.5,0.71]^{\mathsf T}$. We set $T = 0.16$~s and $A=15\lambda$. As illustrated in Fig.~\ref{FIG5}, the proposed MA trajectory forms a virtual aperture of approximately $12\lambda \times 14\lambda$. By comparison, the circular trajectory with radius $R^{\rm cir}=\frac{\Delta}{2\sin\left(\tfrac{\pi}{N}\right)}=5\lambda$ yields an aperture of $2R^{\rm cir}\times2R^{\rm cir}=10\lambda \times 10\lambda$. The UPG trajectory forms an aperture of only $\sqrt{N}\Delta \times \sqrt{N}\Delta = 0.25\lambda \times 0.25\lambda$, which is substantially smaller than those of the proposed and circular trajectories. The FPA-UPA array has an aperture of $1.5\lambda \times 1.5\lambda$, and the FPA-CPA array has an aperture of $4\lambda \times 4\lambda$. As a result, the proposed MA trajectory achieves a significantly narrower main lobe than the benchmark schemes, leading to improved spatial resolution and lower direction vector estimation error when the SNR is sufficiently large.

\begin{figure}[!t]
	\centering
	\includegraphics[width=70mm]{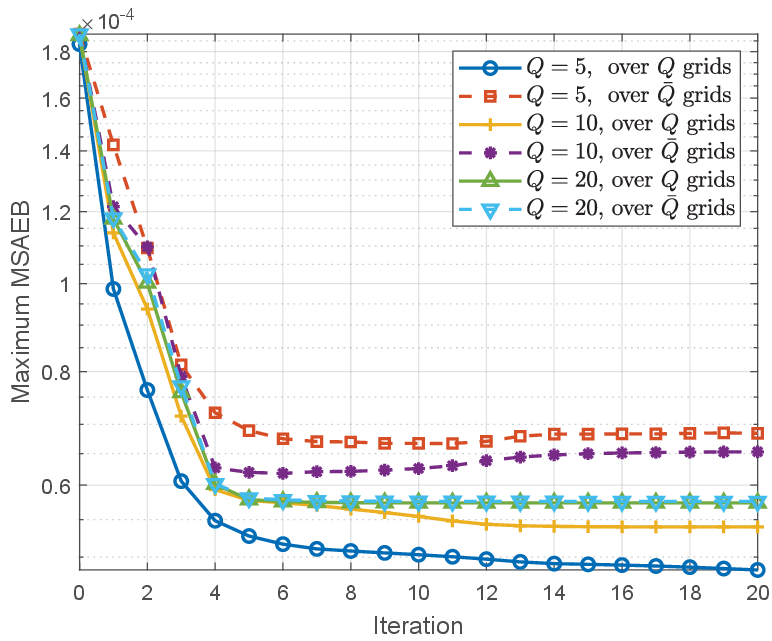}
	\caption{Convergence behavior of Algorithm~\ref{alg1} with different $Q$, where $\bar{Q}=1000$.}
	\label{FIG8}
\end{figure}

\begin{figure}[!t]
	\centering
	\subfigure[Without the region size constraint]{
		\begin{minipage}{.47\textwidth}
			\centering
			\includegraphics[scale=.48]{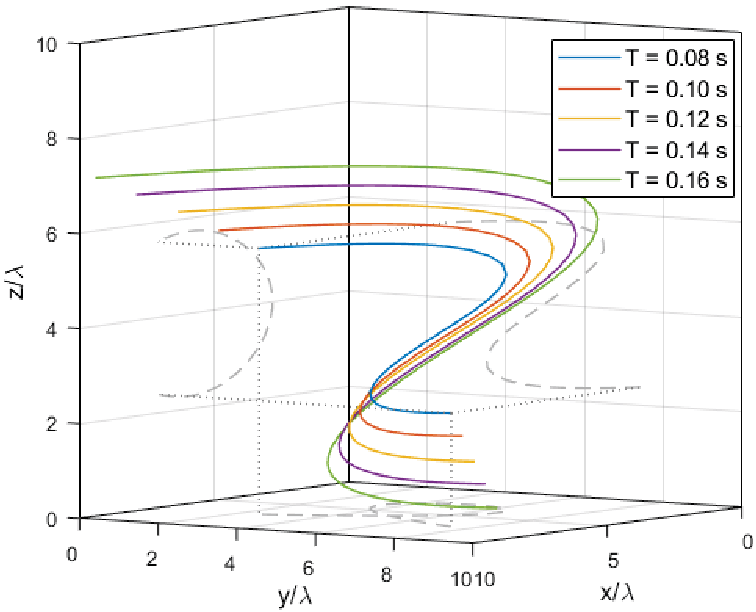}
		\end{minipage}
		\label{Tra_opt_infinite}
	}
	\subfigure[With the region size constraint]{
		\begin{minipage}{.47\textwidth}
			\centering
			\includegraphics[scale=.48]{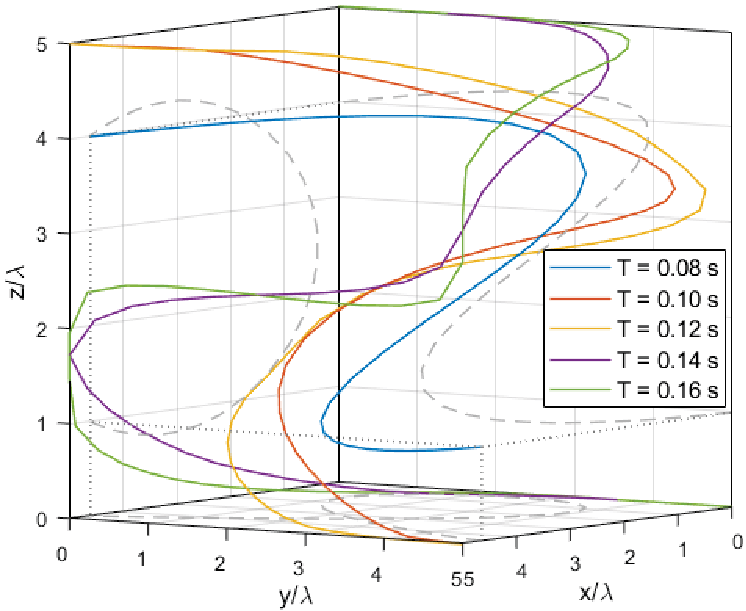}
		\end{minipage}
		\label{Tra_opt_finite}
	}
	\caption{Illustration of MA trajectories when $\mathbb{D}$ is a continuous set.}
	\label{FIG9}
\end{figure}

\begin{figure}[!t]
	\centering
	\includegraphics[width=70mm]{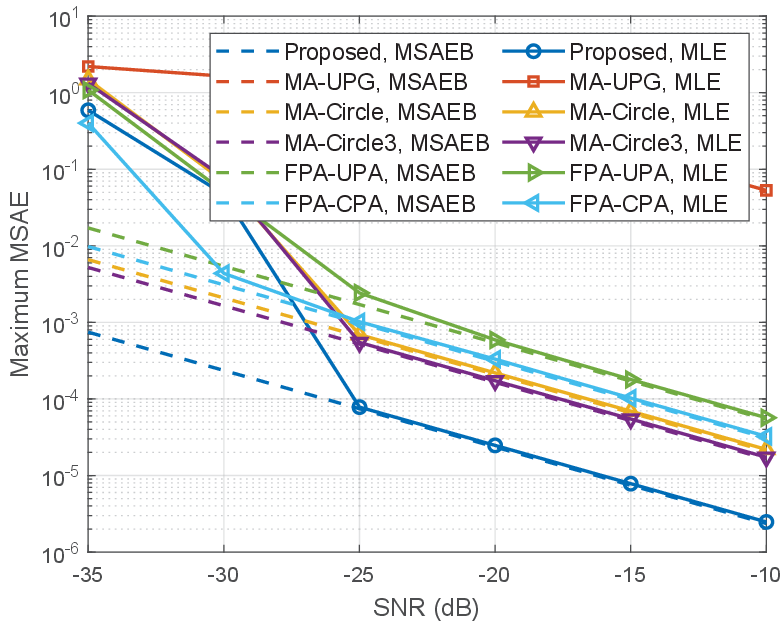}
	\caption{Maximum MSAE versus SNR for different MA trajectories when $\mathbb{D}$ is a continuous set.}
	\label{FIG10}
\end{figure}

\begin{figure}[!t]
	\centering
	\includegraphics[width=70mm]{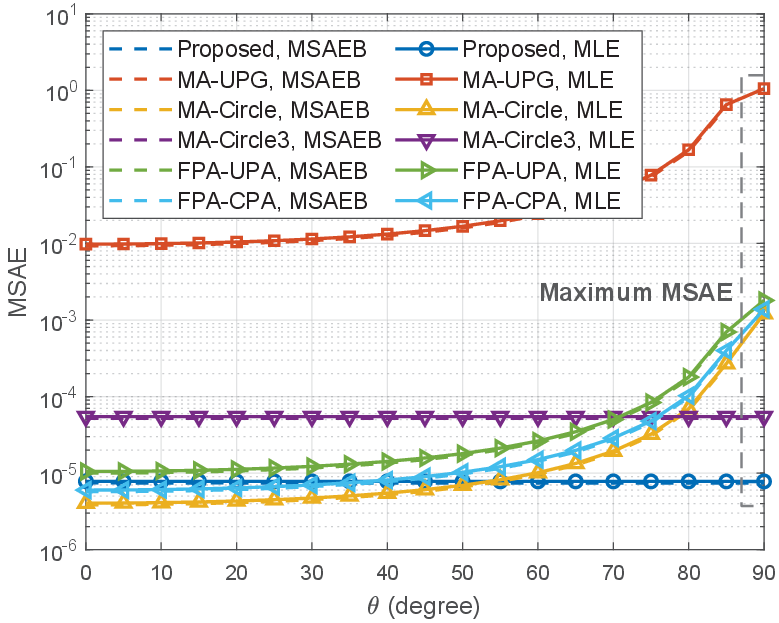}
	\caption{MSAE versus $\theta$ for different MA trajectories when $\mathbb{D}$ is a continuous set.}
	\label{FIG11}
\end{figure}

\subsection{$\mathbb{D}$ is a Continuous Set}

Next, we evaluate the worst-case sensing performance over a continuous angular region $\mathbb{D} = [0, 80^\circ] \times [0, 360^\circ]$ to achieve robust sensing for target direction.  We apply Algorithm~\ref{alg1} for problem (P1) to obtain the proposed MA trajectory.

Fig.~\ref{FIG8} illustrates the convergence behavior of the proposed Algorithm~\ref{alg1} for different numbers of discrete grid points $Q$. We plot both the maximum (i.e., worst-case) MSAEB evaluated over the $Q$ optimization grids and the true maximum MSAEB evaluated over an ultra-dense set of $\bar{Q}=1000$ grid points. Since the proposed Algorithm~\ref{alg1} performs optimization over only $Q$ grid points, convergence is guaranteed for the objective evaluated on these $Q$ grids, whereas convergence over the ultra-dense set of $\bar{Q}$ grids is not guaranteed. When the grid size is small (e.g., $Q=5$), a significant performance gap is observed between the optimized grids and the ultra-dense evaluation. As $Q$ increases to $20$, this gap is substantially reduced, indicating that discretizing $\mathbb{D}$ into a moderate number of grid points is sufficient to ensure robustness over the continuous angular region. Moreover, for $Q=20$, the algorithm converges within $8$ iterations, demonstrating the effectiveness and fast convergence of Algorithm~\ref{alg1}. Thus, $Q=20$ is adopted for the simulation results in the following. 

Fig.~\ref{FIG9} illustrates the optimized MA trajectories under different total sensing time $T$. In Fig.~\ref{Tra_opt_infinite}, the antenna movement region is assumed to be sufficiently large such that constraint \eqref{P1-D} is inactive. In contrast, Fig.~\ref{Tra_opt_finite} considers a finite movement region with size $A = 5\lambda$. It is observed that the optimized MA trajectory expands along the $x$-, $y$-, and $z$-axes to ensure robust direction vector estimation over the continuous angular region $\mathbb{D}$. Moreover, as $T$ increases, the MA trajectory progressively expands within the antenna movement region $\mathcal{C}$, thereby enlarging the effective virtual aperture and enhancing the direction vector estimation accuracy.

Fig.~\ref{FIG10} compares the maximum MSAE versus SNR for different schemes. As shown in Theorem~\ref{theorem_xy_plane}, for the benchmark schemes, the maximum MSAE occurs at the largest elevation angle, i.e., $\theta = 80^\circ$. For the proposed MA trajectory, the maximum MSAE is evaluated over $\bar{Q}$ discrete grid points. As observed in Fig.~\ref{FIG10}, the MLE curves closely match the theoretical MSAEB for all schemes when the SNR is sufficiently high. Moreover, the proposed trajectory significantly outperforms all benchmark schemes, highlighting the necessity of 3-D MA trajectory optimization for robust wide-angle sensing.

Finally, Fig.~\ref{FIG11} shows the MSAE versus the target elevation angle $\theta$, with the SNR fixed as $-15$~dB and $\phi = 0^\circ$. As $\theta$ approaches the endfire direction of the movement plane (i.e., $\theta \to 90^\circ$), the MSAE of the UPA and circular trajectories as well as FPA arrays increases dramatically, which is consistent with the theoretical result in Theorem~\ref{theorem_xy_plane}. In contrast, the 3-D circular and proposed trajectories exhibit flat MSAE curves across the entire angle range. This demonstrates that fully exploiting the 3-D antenna movement can achieve uniform high-accuracy direction estimation over the entire angular region.

	\section{Conclusion}
	In this paper, a novel wireless sensing system was presented where an MA continuously moves and receives sensing signals within a 3-D region to enhance sensing performance compared with conventional FPA-based sensing. We demonstrated that the performance of direction vector estimation for a target is fundamentally related to the 3-D MA trajectory in terms of the MSAEB, which was adopted as a coordinate-invariant performance metric. In particular, the closed-form expression of the MSAEB was derived as a function of the trajectory covariance matrix. Theoretical analysis showed that 2-D antenna movement suffers from performance divergence for target direction close to the endfire direction of the 2-D MA plane, whereas 3-D movement achieves isotropic sensing performance over the entire angular region. To achieve robust sensing performance, a min-max optimization problem was formulated to minimize the worst-case MSAEB over a given continuous angular region wherein the target is located. An efficient SCA algorithm was developed to optimize the 3-D MA trajectory and obtain a locally optimal solution. Numerical results confirmed that the proposed 3-D MA sensing scheme is able to significantly reduce the worst-case MSAE compared with conventional arrays with FPAs and MA systems with 2-D movement only, thus achieving more accurate and robust direction estimation over the entire angular region.

	\appendix
	\subsection{Proof of Theorem 1}
	Define the unknown parameter vector as	$\bm{\kappa} \triangleq [\theta,\phi,\Re(\tilde{\beta}),\Im(\tilde{\beta})]^{\mathsf T} \in \mathbb{R}^{4\times1}$,	which comprises the physical elevation and azimuth AoAs as well as the real and imaginary components of the complex channel gain. For convenience in deriving the CRB, we define	$\bm{u}(\bm{\chi},\bm{\beta}) \triangleq \tilde{\beta}\bm{\alpha}(\bm{R},\bm{\eta}) \in \mathbb{C}^{N\times1}$, 	where $\bm{\beta} \triangleq [\Re(\tilde{\beta}),\Im(\tilde{\beta})]^{\mathsf T}$ contains the real-valued channel coefficients. Let $\bm{F} \in \mathbb{R}^{4\times4}$ denote the Fisher information matrix (FIM) associated with the estimation of $\bm{\kappa}$. The $(p,q)$th entry of $\bm{F}$, for $p,q = 1,2,3,4$, is given by
	\begin{align}
		\bm{F}[p,q] &= 2\Re\left( \frac{\partial \bm{u}(\bm{\chi},\bm{\beta})^{\mathsf H}}{\partial \bm{\kappa}[p]} \bm{R}_{\bm{n}}^{-1} \frac{\partial \bm{u}(\bm{\chi},\bm{\beta})}{\partial \bm{\kappa}[q]} \right) \notag\\
		&=\frac{2}{\sigma^2}\Re\left( \frac{\partial \bm{u}(\bm{\chi},\bm{\beta})^{\mathsf H}}{\partial \bm{\kappa}[p]} \frac{\partial \bm{u}(\bm{\chi},\bm{\beta})}{\partial \bm{\kappa}[q]} \right),
	\end{align}
	where $\bm{R}_{\bm{n}}=\sigma^2 \bm{I}_N$ is the noise covariance matrix.	Then, the FIM $\bm{F}$ can be expressed in a block-structured form as
	\begin{align}\label{F}
		\bm{F} = \begin{bmatrix}
			\bm{J}_{\bm{\chi},\bm{\chi}}  & \bm{J}_{\bm{\chi},\bm{\beta}} \\
			\bm{J}_{\bm{\chi},\bm{\beta}}^{\mathsf T}  & \bm{J}_{\bm{\beta},\bm{\beta}} 
		\end{bmatrix},
	\end{align}
	where the submatrices $\bm{J}_{\bm{\chi},\bm{\chi}}\in \mathbb{R}^{2 \times 2}$, $\bm{J}_{\bm{\chi},\bm{\beta}}\in \mathbb{R}^{2 \times 2}$, and $\bm{J}_{\bm{\beta},\bm{\beta}}\in \mathbb{R}^{2 \times 2}$ are given by
	\begin{align}\label{J}
		\bm{J}_{\bm{\chi},\bm{\chi}} &= \frac{2}{\sigma^2}\Re\left( \frac{\partial \bm{u}(\bm{\chi},\bm{\beta})^{\mathsf H}}{\partial \bm{\chi}} \frac{\partial \bm{u}(\bm{\chi},\bm{\beta})}{\partial \bm{\chi}} \right) \notag\\
		\bm{J}_{\bm{\chi},\bm{\beta}} &= \frac{2}{\sigma^2}\Re\left( \frac{\partial \bm{u}(\bm{\chi},\bm{\beta})^{\mathsf H}}{\partial \bm{\chi}} \frac{\partial \bm{u}(\bm{\chi},\bm{\beta})}{\partial \bm{\beta}} \right) \notag\\
		\bm{J}_{\bm{\beta},\bm{\beta}} &= \frac{2}{\sigma^2}\Re\left( \frac{\partial \bm{u}(\bm{\chi},\bm{\beta})^{\mathsf H}}{\partial \bm{\beta}} \frac{\partial \bm{u}(\bm{\chi},\bm{\beta})}{\partial \bm{\beta}} \right).
	\end{align}
	Then, by applying the inversion formula for partitioned matrices, the upper-left $2\times2$ principal submatrix of $\bm{F}^{-1}$ can be written as
	\begin{align}\label{Lambda}
		\bm{\Lambda}\triangleq \begin{bmatrix}
			\bm{F}^{-1}[1,1]  & \bm{F}^{-1}[1,2] \\
			\bm{F}^{-1}[2,1]  & \bm{F}^{-1}[2,2] 
		\end{bmatrix} = \left[\bm{J}_{\bm{\chi},\bm{\chi}} - \bm{J}_{\bm{\chi},\bm{\beta}} \bm{J}_{\bm{\beta},\bm{\beta}}^{-1} \bm{J}_{\bm{\chi},\bm{\beta}}^{\mathsf T}\right]^{-1}.
	\end{align}
	Accordingly, the CRBs associated with the elevation and azimuth AoA estimation are obtained as
	\begin{align}
		{\rm CRB}_\theta(\bm{R})=\bm{\Lambda}[1,1], \quad
		{\rm CRB}_\phi(\bm{R})=\bm{\Lambda}[2,2].
	\end{align}
	To simplify the notation, we re-denote $\bm{\alpha}(\bm{R},\bm{\eta})$ as $\bm{\alpha}$ in the sequel, and further define the partial derivative of $\bm{\alpha}$ w.r.t. $\bm{\chi}$ as $\dot{\bm{\alpha}}_{\bm{\chi}} \triangleq \frac{\partial \bm{\alpha}}{\partial \bm{\chi}} = \left[\dot{\bm{\alpha}}_\theta, \dot{\bm{\alpha}}_\phi\right] \in \mathbb{C}^{N\times 2}$,
	where the partial derivatives w.r.t. $\theta$ and $\phi$ are given by
	\begin{align}
		\dot{\bm{\alpha}}_\theta &\triangleq \frac{\partial \bm{\alpha}}{\partial \theta} = j\frac{2\pi}{\lambda} \textrm{diag}(\tilde{\bm{f}})\bm{\alpha} \in \mathbb{C}^{N\times1}, \\
		\dot{\bm{\alpha}}_\phi &\triangleq \frac{\partial \bm{\alpha}}{\partial \phi} = j\frac{2\pi}{\lambda} \textrm{diag}(\tilde{\bm{g}})\bm{\alpha} \in \mathbb{C}^{N\times1},\notag
	\end{align}
	with
	\begin{align}
		\tilde{\bm{f}}\triangleq \bm{R}^{\mathsf T}\bm{f} \in \mathbb{R}^{N\times1}, \quad
		\tilde{\bm{g}}\triangleq \sin\theta \bm{R}^{\mathsf T}\bm{g} \in \mathbb{R}^{N\times1}.
	\end{align}
	Based on the above definitions, the derivative of $\bm{u}(\bm{\chi},\bm{\beta})$ w.r.t. $\bm{\chi}$ can be compactly written as $\frac{\partial \bm{u}(\bm{\chi},\bm{\beta})}{\partial \bm{\chi}} = \tilde{\beta}\left[\dot{\bm{\alpha}}_\theta,  \dot{\bm{\alpha}}_\phi \right]$. 
	Moreover, differentiating $\bm{u}(\bm{\chi},\bm{\beta})$ w.r.t. $\bm{\beta}$ yields
	\begin{align}\label{fbeta}
		\frac{\partial \bm{u}(\bm{\chi},\bm{\beta})}{\partial \bm{\beta}} = \frac{\partial \tilde{\beta}\bm{\alpha}}{\partial \bm{\beta}} = \bm{\alpha} [1,j].
	\end{align}
	Then, $\bm{J}_{\bm{\chi},\bm{\chi}}$ in \eqref{J} can be further written as
	\begin{align}
		\bm{J}_{\bm{\chi},\bm{\chi}} &\overset{(a)}= \frac{2}{\sigma^2}\Re\left( \left(\frac{\partial \bm{u}(\bm{\chi},\bm{\beta})}{\partial \bm{\chi}}\right)^{\mathsf H} \frac{\partial \bm{u}(\bm{\chi},\bm{\beta})}{\partial \bm{\chi}} \right) \notag\\
		&\triangleq \frac{2|\tilde{\beta}|^2}{\sigma^2}\Re\left( \begin{bmatrix}
			Q_{\theta,\theta} & Q_{\theta,\phi} \\
			Q_{\phi,\theta}  & Q_{\phi,\phi}
		\end{bmatrix} \right),
	\end{align}
	where the equality $(a)$ holds because $\frac{\partial \bm{u}(\bm{\chi},\bm{\beta})^{\mathsf H}}{\partial \bm{\chi}} = \left(\frac{\partial \bm{u}(\bm{\chi},\bm{\beta})}{\partial \bm{\chi}}\right)^{\mathsf H}$. $Q_{a,b}\triangleq \dot{\bm{\alpha}}_a^{\mathsf H} \dot{\bm{\alpha}}_b$, $a,b\in\{\theta,\phi\}$, which can be further written as
	\begin{align}\label{Qab}
		Q_{\theta,\theta} &= \frac{4\pi^2}{\lambda^2} \bm{\alpha}^{\mathsf H}\textrm{diag}(\tilde{\bm{f}})\textrm{diag}(\tilde{\bm{f}})\bm{\alpha} \\
		&= \frac{4\pi^2}{\lambda^2} \bm{\alpha}^{\mathsf H}\textrm{diag}(|\tilde{\bm{f}}|^2)\bm{\alpha} \overset{(b)}= \frac{4\pi^2}{\lambda^2} \|\tilde{\bm{f}}\|_2^2, \notag
	\end{align}
	where the equality $(b)$ holds because $|\bm{\alpha}[n]|=1$, $n=1,2,\ldots,N$. Similarly, $Q_{\theta,\phi}$, $Q_{\phi,\theta}$, and $Q_{\phi,\phi}$ can be further represented as
	\begin{align}
		Q_{\theta,\phi} = Q_{\phi,\theta} = \frac{4\pi^2}{\lambda^2} \tilde{\bm{f}}^{\mathsf T} \tilde{\bm{g}}, \quad
		Q_{\phi,\phi} = \frac{4\pi^2}{\lambda^2} \|\tilde{\bm{g}}\|_2^2.
	\end{align}	
	Then, $\bm{J}_{\bm{\chi},\bm{\chi}}$ can be simplified as
	\begin{align}
		\bm{J}_{\bm{\chi},\bm{\chi}} &= \frac{8\pi^2|\tilde{\beta}|^2}{\sigma^2\lambda^2}\begin{bmatrix}
			\|\tilde{\bm{f}}\|_2^2 & \tilde{\bm{f}}^{\mathsf T} \tilde{\bm{g}} \\
			\tilde{\bm{f}}^{\mathsf T} \tilde{\bm{g}} & \|\tilde{\bm{g}}\|_2^2
		\end{bmatrix}.
	\end{align}
	Moreover, $\bm{J}_{\bm{\chi},\bm{\beta}}$ can be further written as
	\begin{align}
		\bm{J}_{\bm{\chi},\bm{\beta}} &= \frac{2}{\sigma^2}\Re\left( \frac{\partial \bm{u}(\bm{\chi},\bm{\beta})^{\mathsf H}}{\partial \bm{\chi}} \frac{\partial \bm{u}(\bm{\chi},\bm{\beta})}{\partial \bm{\beta}} \right) \notag\\
		&= \frac{2}{\sigma^2}\Re\left(\begin{bmatrix}
			\tilde{\beta}^{\mathsf *}\dot{\bm{\alpha}}_\theta^{\mathsf H} \bm{\alpha} [1,j] \\
			\tilde{\beta}^{\mathsf *}\dot{\bm{\alpha}}_\phi^{\mathsf H} \bm{\alpha} [1,j]
		\end{bmatrix}  \right),
	\end{align}
	where $\dot{\bm{\alpha}}_\theta^{\mathsf H} \bm{\alpha}$ can be derived following steps analogous to those used in \eqref{Qab} as
	\begin{align}
		\dot{\bm{\alpha}}_\theta^{\mathsf H} \bm{\alpha} &= -j\frac{2\pi}{\lambda} \bm{\alpha}^{\mathsf H} \textrm{diag}(\tilde{\bm{f}})\bm{\alpha} = -j\frac{2\pi}{\lambda} \textrm{sum}(\tilde{\bm{f}}),
	\end{align}
	where $\mathrm{sum}(\bm{a})$ denotes the sum of all entries of vector $\bm{a}$. Similarly, $\dot{\bm{\alpha}}_\phi^{\mathsf H} \bm{\alpha}$ can be further written as $\dot{\bm{\alpha}}_\phi^{\mathsf H} \bm{\alpha} = -j\frac{2\pi}{\lambda} \textrm{sum}(\tilde{\bm{g}})$.
	Then, $\bm{J}_{\bm{\chi},\bm{\beta}}$ can be simplified as
	\begin{align}
		\bm{J}_{\bm{\chi},\bm{\beta}} &= \frac{4\pi}{\sigma^2\lambda} \Re\left(-j\tilde{\beta}^{\mathsf *}\begin{bmatrix}
			\textrm{sum}(\tilde{\bm{f}}) & j\textrm{sum}(\tilde{\bm{f}}) \\
			 \textrm{sum}(\tilde{\bm{g}}) & j \textrm{sum}(\tilde{\bm{g}})
		\end{bmatrix}  \right) \\
		&=  \frac{4\pi}{\sigma^2\lambda}\begin{bmatrix}
			-\Im(\tilde{\beta})\textrm{sum}(\tilde{\bm{f}}) & \Re(\tilde{\beta})\textrm{sum}(\tilde{\bm{f}}) \\
			-\Im(\tilde{\beta}) \textrm{sum}(\tilde{\bm{g}}) & \Re(\tilde{\beta}) \textrm{sum}(\tilde{\bm{g}})
		\end{bmatrix}  \triangleq  \frac{4\pi}{\sigma^2\lambda}\bm{\zeta}\bar{\bm{\beta}}^{\mathsf T},  \notag
	\end{align}
	where $\bm{\zeta}\triangleq[\textrm{sum}(\tilde{\bm{f}}), \textrm{sum}(\tilde{\bm{g}})]^{\mathsf T}$ and $\bar{\bm{\beta}}\triangleq[-\Im(\tilde{\beta}),\Re(\tilde{\beta})]^{\mathsf T}$.
	
	Next, substituting \eqref{fbeta} into \eqref{J}, $\bm{J}_{\bm{\beta},\bm{\beta}}$ can be derived following steps similar to those in \eqref{Qab} as
	\begin{align}
		\bm{J}_{\bm{\beta},\bm{\beta}} &= \frac{2}{\sigma^2}\Re\left( \frac{\partial \bm{u}(\bm{\chi},\bm{\beta})^{\mathsf H}}{\partial \bm{\beta}} \frac{\partial \bm{u}(\bm{\chi},\bm{\beta})}{\partial \bm{\beta}} \right) \notag\\
		&= \frac{2}{\sigma^2}\Re\left( [1,j]^{\mathsf H}\bm{\alpha}^{\mathsf H}  \bm{\alpha} [1,j] \right) \notag\\
		&= \frac{2N}{\sigma^2}\Re\left( \begin{bmatrix}
			1  & j \\
			-j  & 1
		\end{bmatrix} \right) = \frac{2N}{\sigma^2}\bm{I}_2.
	\end{align}
	Then, $\bm{\Lambda}$ in \eqref{Lambda} can be further simplified as
	\begin{align}
		\bm{\Lambda} &= \left[\bm{J}_{\bm{\chi},\bm{\chi}} - \bm{J}_{\bm{\chi},\bm{\beta}} \bm{J}_{\bm{\beta},\bm{\beta}}^{-1} \bm{J}_{\bm{\chi},\bm{\beta}}^{\mathsf T}\right]^{-1} \\
		&= \Bigg[\frac{8\pi^2|\tilde{\beta}|^2}{\sigma^2\lambda^2}\begin{bmatrix}
			\|\tilde{\bm{f}}\|_2^2 &  \tilde{\bm{f}}^{\mathsf T} \tilde{\bm{g}} \\
			 \tilde{\bm{f}}^{\mathsf T} \tilde{\bm{g}} & \|\tilde{\bm{g}}\|_2^2
		\end{bmatrix} \notag\\
		&~~~~~~ - \frac{\sigma^2}{2N}\left(\frac{4\pi}{\sigma^2\lambda}\right)^2\bm{\zeta}\bar{\bm{\beta}}^{\mathsf T} \bar{\bm{\beta}}\bm{\zeta}^{\mathsf T}\Bigg]^{-1} \notag\\
		&= \frac{\sigma^2\lambda^2N}{8\pi^2|\tilde{\beta}|^2} \Bigg[\begin{bmatrix}
			N\|\tilde{\bm{f}}\|_2^2 & N \tilde{\bm{f}}^{\mathsf T} \tilde{\bm{g}} \\
			N \tilde{\bm{f}}^{\mathsf T} \tilde{\bm{g}} & N\|\tilde{\bm{g}}\|_2^2
		\end{bmatrix} \notag\\
		&~~~~~~ - \begin{bmatrix}
			\textrm{sum}(\tilde{\bm{f}})\textrm{sum}(\tilde{\bm{f}}) & \textrm{sum}(\tilde{\bm{f}})  \textrm{sum}(\tilde{\bm{g}}) \\
			\textrm{sum}(\tilde{\bm{f}})  \textrm{sum}(\tilde{\bm{g}}) &  \textrm{sum}(\tilde{\bm{g}})  \textrm{sum}(\tilde{\bm{g}})
		\end{bmatrix}\Bigg]^{-1} \notag\\
		&\overset{(c_1)}= \frac{\sigma^2\lambda^2N}{8\pi^2|\tilde{\beta}|^2} \begin{bmatrix}
			N^2{\rm var}(\tilde{\bm{f}}) & N^2{\rm cov}(\tilde{\bm{f}}, \tilde{\bm{g}}) \\
			N^2{\rm cov}(\tilde{\bm{f}}, \tilde{\bm{g}}) & N^2{\rm var}(\tilde{\bm{g}})
		\end{bmatrix} ^{-1} \notag\\
		&\overset{(c_2)}= \frac{\sigma^2\lambda^2}{8\pi^2PN|\beta|^2} \frac{1}{{\rm var}(\tilde{\bm{f}}){\rm var}(\tilde{\bm{g}})-{\rm cov}(\tilde{\bm{f}},\tilde{\bm{g}})^2} \notag\\
		&~~~~~~\begin{bmatrix}
			{\rm var}(\tilde{\bm{g}}) & -{\rm cov}(\tilde{\bm{f}},\tilde{\bm{g}}) \\
			-{\rm cov}(\tilde{\bm{f}},\tilde{\bm{g}}) & {\rm var}(\tilde{\bm{f}})
		\end{bmatrix},\notag
	\end{align}
	where the equality $(c_1)$ holds because of the definitions of the variance function ${\rm var}(\tilde{\bm{f}})$ and the covariance function ${\rm cov}(\tilde{\bm{f}},\tilde{\bm{g}})$ in \eqref{U}, the equality $(c_2)$ holds since $\tilde{\beta} = \beta s$ and using the inversion formula for a $2\times2$ matrix, i.e., $\begin{bmatrix}
		a & b \\
		c & d
	\end{bmatrix}^{-1}=\frac{1}{ad-bc}\begin{bmatrix}
		d & -b \\
		-c & a
	\end{bmatrix}$. Moreover, since $\|\tilde{\bm{f}}\|_2^2 = \tilde{\bm{f}}^{\mathsf T}\bm{I}_N\tilde{\bm{f}}$ and $\textrm{sum}(\tilde{\bm{f}}) = \tilde{\bm{f}}^{\mathsf T}\bm{1}_N$, we can equivalently rewrite ${\rm var}(\tilde{\bm{f}})$, ${\rm var}(\tilde{\bm{g}})$, and ${\rm cov}(\tilde{\bm{f}},\tilde{\bm{g}})$ into quadratic forms as
	\begin{align}
	{\rm var}(\tilde{\bm{f}}) &= \tilde{\bm{f}}^{\mathsf T} \bm{B} \tilde{\bm{f}} = \bm{f}^{\mathsf T}\bm{R} \bm{B} \bm{R}^{\mathsf T}\bm{f} = \bm{f}^{\mathsf T} \bm{U} \bm{f}, \\
	{\rm var}(\tilde{\bm{g}}) &= \tilde{\bm{g}}^{\mathsf T} \bm{B} \tilde{\bm{g}} = (\sin\theta)^2 \bm{g}^{\mathsf T}\bm{R} \bm{B} \bm{R}^{\mathsf T}\bm{g} = (\sin\theta)^2 \bm{g}^{\mathsf T} \bm{U} \bm{g}, \notag\\
	{\rm cov}(\tilde{\bm{f}},\tilde{\bm{g}}) &= \tilde{\bm{f}}^{\mathsf T} \bm{B} \tilde{\bm{g}} = \sin\theta \bm{f}^{\mathsf T}\bm{R} \bm{B} \bm{R}^{\mathsf T}\bm{g} = \sin\theta \bm{f}^{\mathsf T} \bm{U} \bm{g}, \notag
	\end{align}
	where $\bm{B}$ defined in \eqref{B} is a PSD matrix.	Finally, the CRB expression of $\bm{\chi}$ can be written as
	\begin{align}\label{CRB_app}
		&{\rm CRB}_\theta(\bm{R})=\bm{\Lambda}[1,1] \\
		&~~= \frac{\sigma^2\lambda^2}{8\pi^2PN|\beta|^2} \frac{{\rm var}(\tilde{\bm{g}})}{{\rm var}(\tilde{\bm{f}}){\rm var}(\tilde{\bm{g}})-{\rm cov}(\tilde{\bm{f}},\tilde{\bm{g}})^2} \notag\\
		&~~=\frac{\sigma^2\lambda^2}{8\pi^2PN|\beta|^2} \frac{\bm{g}^{\mathsf T} \bm{U} \bm{g}}{(\bm{f}^{\mathsf T} \bm{U} \bm{f})(\bm{g}^{\mathsf T} \bm{U} \bm{g})-(\bm{f}^{\mathsf T} \bm{U} \bm{g})^2},\notag\\
		&{\rm CRB}_\phi(\bm{R})=\bm{\Lambda}[2,2] \notag\\
		&~~= \frac{\sigma^2\lambda^2}{8\pi^2PN|\beta|^2} \frac{{\rm var}(\tilde{\bm{f}})}{{\rm var}(\tilde{\bm{f}}){\rm var}(\tilde{\bm{g}})-{\rm cov}(\tilde{\bm{f}},\tilde{\bm{g}})^2} \notag\\
		&~~= \frac{\sigma^2\lambda^2}{8\pi^2PN|\beta|^2} \frac{1}{(\sin\theta)^2} \frac{\bm{f}^{\mathsf T} \bm{U} \bm{f}}{(\bm{f}^{\mathsf T} \bm{U} \bm{f})(\bm{g}^{\mathsf T} \bm{U} \bm{g})-(\bm{f}^{\mathsf T} \bm{U} \bm{g})^2}.	\notag
	\end{align}
	This thus completes the proof of Theorem 1.

	\subsection{Proof of Theorem 2}
	
	We first prove the sufficiency. Suppose $\bm{U} = \tau\bm{I}_3$. Since $\bm{f}$ and $\bm{g}$ are orthonormal vectors, i.e., $\|\bm{f}\|_2 = \|\bm{g}\|_2 = 1$ and $\bm{f}^{\mathsf T}\bm{g} = 0$, we have
	\begin{align}\label{fgI3}
		\bm{f}^{\mathsf T} \bm{U} \bm{f} &= \tau \bm{f}^{\mathsf T} \bm{I}_3 \bm{f} = \tau \|\bm{f}\|_2^2 = \tau, \\
		\bm{g}^{\mathsf T} \bm{U} \bm{g} &= \tau \bm{g}^{\mathsf T} \bm{I}_3 \bm{g} = \tau \|\bm{g}\|_2^2 = \tau, \notag\\
		\bm{f}^{\mathsf T} \bm{U} \bm{g} &= \tau \bm{f}^{\mathsf T} \bm{I}_3 \bm{g} = \tau \bm{f}^{\mathsf T} \bm{g} = 0. \notag
	\end{align}
	Substituting \eqref{fgI3} into the definition of ${\rm MSAEB}_{\bm{\chi}}(\bm{R})$ in \eqref{MSAEB}, we obtain ${\rm MSAEB}_{\bm{\chi}}(\bm{R}) = \frac{2\rho}{\tau}$,
	which is a constant independent of the direction $\bm{\chi}$. Thus, the sufficiency is proved.
	
	Next, we prove the necessity. Assume that ${\rm MSAEB}_{\bm{\chi}}(\bm{R})$ is constant for all $\bm{\chi} \in [0,\pi] \times [0, 2\pi]$. Since $\bm{U}$ is a real symmetric matrix, it admits the eigenvalue decomposition $\bm{U} = \bm{E} \bm{\Sigma} \bm{E}^{\mathsf T}$, where $\bm{\Sigma} = \text{diag}([\lambda_1, \lambda_2, \lambda_3]^{\mathsf T})$ contains the eigenvalues, and $\bm{E} = [\bm{e}_1, \bm{e}_2, \bm{e}_3]\in\mathbb{R}^{3 \times 3}$ is an orthogonal matrix containing the corresponding eigenvectors. In the following, we evaluate $F_{\bm{\chi}}(\bm{R})$ along the principal axes defined by the eigenvectors. First, consider the direction aligned with the third eigenvector, i.e., $\bm{\eta} = \bm{e}_3$. Since $\bm{f}^{\mathsf T}\bm{\eta} = \bm{g}^{\mathsf T}\bm{\eta} = \bm{f}^{\mathsf T}\bm{g} = 0$ and $\|\bm{f}\|_2 = \|\bm{g}\|_2 = \|\bm{\eta}\|_2 = 1$, the set $\{\bm{f}, \bm{g}, \bm{\eta}\}$ forms an orthonormal basis. Moreover, since $\{\bm{f}, \bm{g}\}$ and $\{\bm{e}_1, \bm{e}_2\}$ are two orthonormal bases of the same subspace orthogonal to $\bm{\eta} = \bm{e}_3$, $[\bm{f}, \bm{g}]$ can be written as
	\begin{align}
		[\bm{f}, \bm{g}] = \bar{\bm{Q}} [\bm{e}_1, \bm{e}_2],
	\end{align}
	where $\bar{\bm{Q}} = \bm{E} \bm{\Psi} \bm{E}^{\mathsf T}$ is the rotation matrix and $\bm{\Psi} =
	\begin{bmatrix}
		\cos\psi & -\sin\psi & 0 \\
		\sin\psi & \cos\psi & 0 \\
		0 & 0 & 1
	\end{bmatrix}$ is the canonical rotation matrix representing a rotation by an angle $\psi$ w.r.t. the axis $\bm{e}_3$. Thus, we can verify $\bm{\eta}=\bm{e}_3$ as
	\begin{align}
		\bm{\eta} &= \bar{\bm{Q}} \bm{e}_3 = (\bm{E} \bm{\Psi} \bm{E}^{\mathsf T}) \bm{e}_3 = \bm{E} \bm{\Psi} (\bm{E}^{\mathsf T} \bm{e}_3) \\ 
		&= \bm{E} \bm{\Psi} [0,0,1]^{\mathsf T} = \bm{E} [0,0,1]^{\mathsf T} = \bm{e}_3. \notag
	\end{align}
	Moreover, $\bm{f}$ and $\bm{g}$ can be further written as
	\begin{align}
		\bm{f} &= \bar{\bm{Q}} \bm{e}_1 = (\bm{E} \bm{\Psi} \bm{E}^{\mathsf T}) \bm{e}_1 = \bm{E} \bm{\Psi} (\bm{E}^{\mathsf T} \bm{e}_1) \\ 
		&= \bm{E} \bm{\Psi} [1,0,0]^{\mathsf T} = \bm{E} [\cos\psi,\sin\psi,0]^{\mathsf T}\notag\\ 
		&= \cos\psi \bm{e}_1 + \sin\psi \bm{e}_2, \notag
	\end{align}
	\begin{align}
		\bm{g} &= \bar{\bm{Q}} \bm{e}_2 = (\bm{E} \bm{\Psi} \bm{E}^{\mathsf T}) \bm{e}_2 = \bm{E} \bm{\Psi} (\bm{E}^{\mathsf T} \bm{e}_2) \\ 
		&= \bm{E} \bm{\Psi} [0,1,0]^{\mathsf T} = \bm{E} [-\sin\psi,\cos\psi,0]^{\mathsf T}\notag\\ 
		&= -\sin\psi \bm{e}_1 + \cos\psi \bm{e}_2. \notag
	\end{align}
	Then, using the spectral decomposition $\bm{U} = \sum_{i=1}^3 \lambda_i \bm{e}_i \bm{e}_i^{\mathsf T}$, and the orthonormality property $\bm{e}_i^{\mathsf T} \bm{e}_j = 
	\begin{cases} 
		1, & i=j, \\ 
		0, & i\neq j. 
	\end{cases}$, $\bm{f}^{\mathsf T} \bm{U} \bm{f}$, $\bm{g}^{\mathsf T} \bm{U} \bm{g}$, and $\bm{f}^{\mathsf T} \bm{U} \bm{g}$ can be further written as
	\begin{align}
		\bm{f}^{\mathsf T} \bm{U} \bm{f} &= \lambda_1 (\cos\psi)^2 + \lambda_2 (\sin\psi)^2, \\
		\bm{g}^{\mathsf T} \bm{U} \bm{g} &= \lambda_1 (\sin\psi)^2 + \lambda_2 (\cos\psi)^2, \notag\\
		\bm{f}^{\mathsf T} \bm{U} \bm{g} &= (\lambda_2 - \lambda_1)\sin\psi\cos\psi. \notag
	\end{align}
	Therefore, $F_{\bm{\chi}}(\bm{R})\big|_{\bm{\eta}=\bm{e}_3}$ can be simplified as
	\begin{align}
		F_{\bm{\chi}}(\bm{R})\big|_{\bm{\eta}=\bm{e}_3} &= \frac{\bm{g}^{\mathsf T} \bm{U} \bm{g} + \bm{f}^{\mathsf T} \bm{U} \bm{f}}{(\bm{f}^{\mathsf T} \bm{U} \bm{f})(\bm{g}^{\mathsf T} \bm{U} \bm{g})-(\bm{f}^{\mathsf T} \bm{U} \bm{g})^2} \\
		&= \frac{\lambda_1 + \lambda_2}{\lambda_1 \lambda_2} = \frac{1}{\lambda_1} + \frac{1}{\lambda_2}. \notag
	\end{align}
	Similarly, by aligning the target direction $\bm{\eta}$ with $\bm{e}_1$ and $\bm{e}_2$, respectively, and following steps analogous to those used to derive $F_{\bm{\chi}}(\bm{R})\big|_{\bm{\eta}=\bm{e}_3}$, we obtain
	\begin{align}
		F_{\bm{\chi}}(\bm{R})\big|_{\bm{\eta}=\bm{e}_1} = \frac{1}{\lambda_2} + \frac{1}{\lambda_3}, \quad
		F_{\bm{\chi}}(\bm{R})\big|_{\bm{\eta}=\bm{e}_2} = \frac{1}{\lambda_1} + \frac{1}{\lambda_3}. 
	\end{align}
	Due to the isotropy of ${\rm MSAEB}_{\bm{\chi}}(\bm{R})$, $F_{\bm{\chi}}(\bm{R})\big|_{\bm{\eta}=\bm{e}_1}$, $F_{\bm{\chi}}(\bm{R})\big|_{\bm{\eta}=\bm{e}_2}$, and $F_{\bm{\chi}}(\bm{R})\big|_{\bm{\eta}=\bm{e}_3}$ must be equal, i.e.,
	\begin{equation}
		\frac{1}{\lambda_1} + \frac{1}{\lambda_2} = \frac{1}{\lambda_2} + \frac{1}{\lambda_3} = \frac{1}{\lambda_1} + \frac{1}{\lambda_3},
	\end{equation}
	where the solution is $\lambda_1 = \lambda_2 = \lambda_3 \triangleq \tau$.
	Consequently, the APV covariance matrix $\bm{U}$ is given by
	\begin{align}
		\bm{U} = \bm{E} (\tau \bm{I}_3) \bm{E}^{\mathsf T} = \tau \bm{E} \bm{E}^{\mathsf T} = \tau \bm{I}_3.
	\end{align}
	This thus completes the proof of Theorem 2.

	\subsection{Proof of Theorem 3}
	
	When the MA trajectory is restricted in the $x$-$y$ plane, we have ${\rm var}(\bm{z}) = 0$, ${\rm cov}(\bm{x},\bm{z}) = 0$, and ${\rm cov}(\bm{y},\bm{z}) = 0$. Consequently, the APV covariance matrix $\bm{U}$ reduces to
	\begin{align}\label{U_xy_plane}
		\bm{U} = \begin{bmatrix}
			{\rm var}(\bm{x}) & {\rm cov}(\bm{x},\bm{y}) & 0 \\
			{\rm cov}(\bm{x},\bm{y}) & {\rm var}(\bm{y}) & 0 \\
			0 & 0 & 0
		\end{bmatrix}.
	\end{align}
	Substituting \eqref{U_xy_plane} into $\bm{g}^{\mathsf T} \bm{U} \bm{g}$, $\bm{f}^{\mathsf T} \bm{U} \bm{f}$, and $\bm{f}^{\mathsf T} \bm{U} \bm{g}$, we obtain
	\begin{align}
		\bm{g}^{\mathsf T} \bm{U} \bm{g} &= (\sin\phi)^2 {\rm var}(\bm{x}) + (\cos\phi)^2 {\rm var}(\bm{y}) \notag\\
		&~~~~~~ - 2\sin\phi\cos\phi {\rm cov}(\bm{x},\bm{y}), \\
		\bm{f}^{\mathsf T} \bm{U} \bm{f} &= (\cos\theta)^2 \big[ (\cos\phi)^2 {\rm var}(\bm{x}) + (\sin\phi)^2 {\rm var}(\bm{y}) \notag\\
		&~~~~~~ + 2\sin\phi\cos\phi {\rm cov}(\bm{x},\bm{y}) \big], \notag\\
		\bm{f}^{\mathsf T} \bm{U} \bm{g} &= \cos\theta \big[ -\sin\phi\cos\phi {\rm var}(\bm{x}) + \sin\phi\cos\phi {\rm var}(\bm{y}) \notag\\
		&~~~~~~+ ((\cos\phi)^2 - (\sin\phi)^2){\rm cov}(\bm{x},\bm{y}) \big]. \notag
	\end{align}
	For notational brevity, we define three terms that capture the azimuth dependence but are independent of $\theta$:
	\begin{align}
		U_{gg}(\phi) \triangleq \bm{g}^{\mathsf T} \bm{U} \bm{g}, 
		~U_{ff}(\phi) \triangleq \frac{\bm{f}^{\mathsf T} \bm{U} \bm{f}}{(\cos\theta)^2}, 
		~U_{fg}(\phi) \triangleq \frac{\bm{f}^{\mathsf T} \bm{U} \bm{g}}{\cos\theta}. 
	\end{align}
	Then, the function $F_{\bm{\chi}}(\bm{R})$ can be rewritten as
	\begin{align}
		F_{\bm{\chi}}(\bm{R}) &=  \frac{1}{(\cos\theta)^2} \frac{U_{gg}(\phi)}{U_{ff}(\phi) U_{gg}(\phi) - U_{fg}(\phi)^2} \notag\\
		&~~~~+ \frac{U_{ff}(\phi)}{U_{ff}(\phi) U_{gg}(\phi) - U_{fg}(\phi)^2}. \notag
	\end{align}
	Let $A(\phi) \triangleq \frac{U_{gg}(\phi)}{U_{ff}(\phi) U_{gg}(\phi) - U_{fg}(\phi)^2}$ and $B(\phi) \triangleq \frac{U_{ff}(\phi)}{U_{ff}(\phi) U_{gg}(\phi) - U_{fg}(\phi)^2}$. For a fixed azimuth angle $\phi$, both $A(\phi)$ and $B(\phi)$ are strictly positive constants. Thus, $F_{\bm{\chi}}(\bm{R})$ can be simplified as
	\begin{align}
		F_{\bm{\chi}}(\bm{R}) = \frac{A(\phi)}{(\cos\theta)^2} + B(\phi).
	\end{align}
	Since the function $(\cos\theta)^2$ monotonically decreases from $1$ to $0$ for $\theta \in [0, \pi/2)$, $F_{\bm{\chi}}(\bm{R})$ monotonically increases w.r.t. $\theta$ and approaches infinity as $\theta \to \pi/2$. Conversely, for $\theta \in (\pi/2, \pi]$, $F_{\bm{\chi}}(\bm{R})$ monotonically decreases w.r.t. $\theta$. This thus completes the proof of Theorem 3.

	\bibliographystyle{IEEEtran}
	\bibliography{IEEEabrv,IEEEexample}

\begin{thebibliography}{10}
\providecommand{\url}[1]{#1}
\csname url@samestyle\endcsname
\providecommand{\newblock}{\relax}
\providecommand{\bibinfo}[2]{#2}
\providecommand{\BIBentrySTDinterwordspacing}{\spaceskip=0pt\relax}
\providecommand{\BIBentryALTinterwordstretchfactor}{4}
\providecommand{\BIBentryALTinterwordspacing}{\spaceskip=\fontdimen2\font plus
\BIBentryALTinterwordstretchfactor\fontdimen3\font minus
  \fontdimen4\font\relax}
\providecommand{\BIBforeignlanguage}[2]{{%
\expandafter\ifx\csname l@#1\endcsname\relax
\typeout{** WARNING: IEEEtran.bst: No hyphenation pattern has been}%
\typeout{** loaded for the language `#1'. Using the pattern for}%
\typeout{** the default language instead.}%
\else
\language=\csname l@#1\endcsname
\fi
#2}}
\providecommand{\BIBdecl}{\relax}
\BIBdecl

\bibitem{jiang2021the}
W.~Jiang, B.~Han, M.~A. Habibi, and H.~D. Schotten, ``{The road towards 6G: A
  comprehensive survey},'' \emph{{IEEE} Open J. Commun. Soc.}, vol.~2, pp.
  334--366, Feb. 2021.

\bibitem{mailloux2005phased}
R.~J. Mailloux, \emph{{Phased Array Antenna Handbook}}.\hskip 1em plus 0.5em
  minus 0.4em\relax 2nd ed. Norwood, MA, USA: Artech House, 2005.

\bibitem{roberts2011sparse}
W.~Roberts, L.~Xu, J.~Li, and P.~Stoica, ``{Sparse antenna array design for
  MIMO active sensing applications},'' \emph{{IEEE} Trans. Antennas Propagat.},
  vol.~59, no.~3, pp. 846--858, Mar. 2011.

\bibitem{zhu2023MAMag}
L.~Zhu, W.~Ma, and R.~Zhang, ``Movable antennas for wireless communication:
  Opportunities and challenges,'' \emph{IEEE Commun. Mag.}, vol.~62, no.~6, pp.
  114--120, Jun. 2024.

\bibitem{zhu2025tutorial}
L.~Zhu, W.~Ma, W.~Mei, Y.~Zeng, Q.~Wu, B.~Ning, Z.~Xiao, X.~Shao, J.~Zhang, and
  R.~Zhang, ``A tutorial on movable antennas for wireless networks,''
  \emph{{IEEE} Commun. Surveys Tuts.}, vol.~28, pp. 3002--3054, Feb. 2025.

\bibitem{zhao2009single}
S.~Zhao, H.~Yang, and H.~Yang, ``Single antenna spatial diversity,'' in
  \emph{Proc. IEEE Int. Conf. Wireless Commun., Netw., Mobile Comput. (WiCOM)},
  Beijing, China, Sep. 2009, pp. 1--4.

\bibitem{zhu2022MAmodel}
{L. Zhu, W. Ma, and R. Zhang}, ``Modeling and performance analysis for movable
  antenna enabled wireless communications,'' \emph{IEEE Trans. Wireless
  Commun.}, vol.~23, no.~6, pp. 6234--6250, Jun. 2024.

\bibitem{mei2024movable}
W.~Mei, X.~Wei, B.~Ning, Z.~Chen, and R.~Zhang, ``Movable-antenna position
  optimization: A graph-based approach,'' \emph{IEEE Wireless Commun. Lett.},
  vol.~13, no.~7, pp. 1853--1857, Jul. 2024.

\bibitem{ning2024movable}
B.~Ning, S.~Yang, Y.~Wu, P.~Wang, W.~Mei, C.~Yuen, and E.~Bj{\"o}rnson,
  ``Movable antenna-enhanced wireless communications: General architectures and
  implementation methods,'' \emph{IEEE Wireless Commun.}, vol.~32, no.~5, pp.
  108--116, Oct. 2025.

\bibitem{tang2024secure}
J.~Tang, C.~Pan, Y.~Zhang, H.~Ren, and K.~Wang, ``Secure {MIMO} communication
  relying on movable antennas,'' \emph{IEEE Trans. Commun.}, vol.~73, no.~4,
  pp. 2159--2175, Apr. 2025.

\bibitem{zhu2023MAmultiuser}
L.~Zhu, W.~Ma, B.~Ning, and R.~Zhang, ``Movable-antenna enhanced multiuser
  communication via antenna position optimization,'' \emph{IEEE Trans. Wireless
  Commun.}, vol.~23, no.~7, pp. 7214--7229, Jul. 2024.

\bibitem{wu2023movable}
Y.~Wu, D.~Xu, D.~W.~K. Ng, W.~Gerstacker, and R.~Schober, ``Movable
  antenna-enhanced multiuser communication: Optimal discrete antenna
  positioning and beamforming,'' in \emph{Proc. IEEE Global Commun. Conf.
  (Globecom)}, Kuala Lumpur, Malaysia, Dec. 2023, pp. 7508--7513.

\bibitem{qin2024antenna}
H.~Qin, W.~Chen, Z.~Li, Q.~Wu, N.~Cheng, and F.~Chen, ``Antenna positioning and
  beamforming design for fluid antenna-assisted multi-user downlink
  communications,'' \emph{IEEE Wireless Commun. Lett.}, vol.~13, no.~4, pp.
  1073--1077, Apr. 2024.

\bibitem{cheng2023sum}
Z.~Cheng, N.~Li, J.~Zhu, X.~She, C.~Ouyang, and P.~Chen, ``Sum-rate
  maximization for movable antenna enabled multiuser communications,''
  \emph{arXiv preprint arXiv:2309.11135}, 2023.

\bibitem{yang2024flexible}
S.~Yang, J.~An, Y.~Xiu, W.~Lyu, B.~Ning, Z.~Zhang, M.~Debbah, and C.~Yuen,
  ``Flexible antenna arrays for wireless communications: Modeling and
  performance evaluation,'' \emph{IEEE Trans. Wireless Commun.}, vol.~24,
  no.~6, pp. 4937--4951, Jun. 2025.

\bibitem{hu2024power}
G.~Hu, Q.~Wu, K.~Xu, J.~Ouyang, J.~Si, Y.~Cai, and N.~Al-Dhahir, ``Fluid
  antennas-enabled multiuser uplink: A low-complexity gradient descent for
  total transmit power minimization,'' \emph{IEEE Commun. Lett.}, vol.~28,
  no.~3, pp. 602--606, Mar. 2025.

\bibitem{li2024minimizing}
Q.~Li, W.~Mei, B.~Ning, and R.~Zhang, ``Minimizing movement delay for movable
  antennas via trajectory optimization,'' in \emph{Proc. IEEE Global Commun.
  Conf. (Globecom)}, Cape Town, South Africa, Dec. 2024, pp. 1--6.

\bibitem{ma2022MAmimo}
W.~Ma, L.~Zhu, and R.~Zhang, ``{MIMO} capacity characterization for movable
  antenna systems,'' \emph{IEEE Trans. Wireless Commun.}, vol.~23, no.~4, pp.
  3392--3407, Apr. 2024.

\bibitem{chen2023joint}
X.~Chen, B.~Feng, Y.~Wu, D.~W.~K. Ng, and R.~Schober, ``Joint beamforming and
  antenna movement design for moveable antenna systems based on statistical
  {CSI},'' in \emph{Proc. IEEE Global Commun. Conf. (Globecom)}, Kuala Lumpur,
  Malaysia, Dec. 2023, pp. 4387--4392.

\bibitem{yeyuqi2023fluid}
Y.~Ye, L.~You, J.~Wang, H.~Xu, K.-K. Wong, and X.~Gao, ``{Fluid
  antenna-assisted MIMO transmission exploiting statistical CSI},'' \emph{IEEE
  Commun. Lett.}, vol.~28, no.~1, pp. 223--227, Jan. 2024.

\bibitem{ma2023MAestimation}
W.~Ma, L.~Zhu, and R.~Zhang, ``Compressed sensing based channel estimation for
  movable antenna communications,'' \emph{IEEE Commun. Lett.}, vol.~27, no.~10,
  pp. 2747--2751, Oct. 2023.

\bibitem{xiao2023channel}
Z.~Xiao, S.~Cao, L.~Zhu, Y.~Liu, B.~Ning, X.-G. Xia, and R.~Zhang, ``Channel
  estimation for movable antenna communication systems: A framework based on
  compressed sensing,'' \emph{IEEE Trans. Wireless Commun.}, vol.~23, no.~9,
  pp. 11\,814--11\,830, Sep. 2024.

\bibitem{zhu2023MAarray}
L.~Zhu, W.~Ma, and R.~Zhang, ``Movable-antenna array enhanced beamforming:
  Achieving full array gain with null steering,'' \emph{IEEE Commun. Lett.},
  vol.~27, no.~12, pp. 3340--3344, Dec. 2023.

\bibitem{ma2024multi}
W.~Ma, L.~Zhu, and R.~Zhang, ``Multi-beam forming with movable-antenna array,''
  \emph{IEEE Commun. Lett.}, vol.~28, no.~3, pp. 697--701, Mar. 2024.

\bibitem{ZhuLP_satellite_MA}
L.~Zhu, X.~Pi, W.~Ma, Z.~Xiao, and R.~Zhang, ``Dynamic beam coverage for
  satellite communications aided by movable-antenna array,'' \emph{IEEE Trans.
  Wireless Commun.}, vol.~24, no.~3, pp. 1916--1933, Mar. 2025.

\bibitem{zhu2024nearfield}
{L. Zhu, W. Ma, Z. Xiao, and R. Zhang}, ``Movable antenna enabled near-field
  communications: Channel modeling and performance optimization,'' \emph{IEEE
  Trans. Commun.}, vol.~73, no.~9, pp. 7240--7256, Sep. 2025.

\bibitem{shao20246DMA}
X.~Shao, Q.~Jiang, and R.~Zhang, ``{6D} movable antenna based on user
  distribution: Modeling and optimization,'' \emph{IEEE Trans. Wireless
  Commun.}, vol.~24, no.~1, pp. 355--370, Jan. 2025.

\bibitem{shao2024discrete}
X.~Shao, R.~Zhang, Q.~Jiang, and R.~Schober, ``{6D} movable antenna enhanced
  wireless network via discrete position and rotation optimization,''
  \emph{IEEE J. Select. Areas Commun.}, vol.~43, no.~3, pp. 674--687, Mar.
  2025.

\bibitem{shao2024Mag6DMA}
X.~Shao and R.~Zhang, ``{6DMA} enhanced wireless network with flexible antenna
  position and rotation: Opportunities and challenges,'' \emph{IEEE Commun.
  Mag.}, vol.~63, no.~4, pp. 121--128, Apr. 2025.

\bibitem{shao2024exploiting}
X.~Shao, R.~Zhang, and R.~Schober, ``Exploiting six-dimensional movable antenna
  for wireless sensing,'' \emph{IEEE Wireless Commun. Lett.}, vol.~14, no.~2,
  pp. 265--269, Feb. 2025.

\bibitem{zhuravlev2015experi}
A.~Zhuravlev, V.~Razevig, S.~Ivashov, A.~Bugaev, and M.~Chizh, ``Experimental
  simulation of multi-static radar with a pair of separated movable antennas,''
  in \emph{Proc. IEEE International Conf. Microwaves Commun. Antennas Electron.
  Syst. (COMCAS)}, Nov. 2015, pp. 1--5.

\bibitem{hinske2008using}
S.~Hinske and M.~Langheinrich, ``Using a movable {RFID} antenna to
  automatically determine the position and orientation of objects on a
  tabletop,'' in \emph{Smart Sensing Context, 3rd Eur. Conf.}\hskip 1em plus
  0.5em minus 0.4em\relax Springer, 2008, pp. 14--26.

\bibitem{ma2024MAsensing}
W.~Ma, L.~Zhu, and R.~Zhang, ``Movable antenna enhanced wireless sensing via
  antenna position optimization,'' \emph{IEEE Trans. Wireless Commun.},
  vol.~23, no.~11, pp. 16\,575--16\,589, Nov. 2024.

\bibitem{chen2025MAISACopt}
L.~Chen, M.-M. Zhao, M.-J. Zhao, and R.~Zhang, ``Antenna position and
  beamforming optimization for movable antenna enabled {ISAC}: Optimal
  solutions and efficient algorithms,'' \emph{{IEEE} Trans. Signal Processing},
  vol.~73, pp. 3812--3828, Jul. 2025.

\bibitem{wang2025MAnearsensing}
Y.~Wang, W.~Mei, X.~Wei, B.~Ning, and Z.~Chen, ``Antenna position optimization
  for movable antenna-empowered near-field sensing,'' \emph{arXiv preprint
  arXiv:2502.03169}, 2025.

\bibitem{mao2025movable}
H.~Mao, L.~Zhu, W.~Ma, Z.~Xiao, X.-G. Xia, and R.~Zhang, ``Movable-antenna
  array enhanced multi-target sensing: {CRB} characterization and
  optimization,'' \emph{arXiv preprint arXiv:2511.18907}, 2025.

\bibitem{ma2025movabletra}
W.~Ma, L.~Zhu, and R.~Zhang, ``Movable-antenna trajectory optimization for
  wireless sensing: {CRB} scaling laws over time and space,'' \emph{arXiv
  preprint arXiv:2509.14905}, 2025.

\bibitem{nehorai1999performance}
A.~Nehorai and M.~Hawkes, ``Performance measures for estimating vector
  systems,'' in \emph{Proc. IEEE Int. Conf. Acoust., Speech, Signal Process.
  (ICASSP)}, Phoenix, AZ, USA, Mar. 1999, pp. 1829--1832.

\bibitem{deits2015computing}
R.~Deits and R.~Tedrake, ``{Computing large convex regions of obstacle-free
  space through semidefinite programming},'' in \emph{Algorithmic Foundations
  of Robotics XI}, vol. 107.\hskip 1em plus 0.5em minus 0.4em\relax Springer,
  Jan. 2015, pp. 109--124.

\bibitem{grantcvx}
M.~Grant and S.~Boyd, \emph{{(Mar. 2014). CVX: MATLAB Software for Disciplined
  Convex Programming, Version 2.1}}, [Online]. Available: http://cvxr.com/cvx.

\bibitem{zheng2017generalized}
W.~Zheng, X.~Zhang, and H.~Zhai, ``Generalized coprime planar array geometry
  for 2-{D} {DOA} estimation,'' \emph{IEEE Commun. Lett.}, vol.~21, no.~5, pp.
  1075--1078, May 2017.

\end{thebibliography}
	
\end{document}